\documentclass[conference]{IEEEtran}

\usepackage[compress]{cite}
\usepackage{color}
\usepackage{times}
\usepackage[noend]{algorithmic}
\usepackage{graphicx}
\usepackage{subfigure}
\usepackage{multirow}
\usepackage{mathtools}

\newcommand{\G}{\mathcal{G}}
\newcommand{\E}{\mathbf{E}}
\newcommand{\W}{\mathbf{W}}
\newcommand{\UW}{\mathbf{\mathcal{W}}}
\newcommand{\Sim}{\mathbf{S}}

\newcommand{\I}{\mathbf{I}}

\newcommand{\F}{\mathbf{F}}
\newcommand{\M}{\mathbf{M}}
\newcommand{\A}{\mathbf{A}}
\newcommand{\X}{\mathbf{X}}
\newcommand{\tabincell}[2]{\begin{tabular}{@{}#1@{}}#2\end{tabular}}
\newcommand{\st}{\scriptsize}

\newtheorem{definition}{Definition}

\newtheorem{lemma}{Lemma}
\newtheorem{theorem}{Theorem}
\newtheorem{corollary}{Corollary}

\begin{document}

\title{SimRank Computation on Uncertain Graphs}

\author{\IEEEauthorblockN{Rong Zhu, Zhaonian Zou, Jianzhong Li}
\IEEEauthorblockA{\textit{Department of Computer Science and Technology, Harbin Institute of Technology,
China}\\
{\tt \{rzhu,znzou,lijzh\}@hit.edu.cn}}
}

\maketitle

\begin{abstract}
SimRank is a similarity measure between vertices in a graph, which has become a fundamental technique in graph analytics. Recently, many algorithms have been proposed for efficient evaluation of SimRank similarities. However, the existing SimRank computation algorithms either overlook uncertainty in graph structures or is based on an unreasonable assumption (Du et al). In this paper, we study SimRank similarities on uncertain graphs based on the possible world model of uncertain graphs. Following the random-walk-based formulation of SimRank on deterministic graphs and the possible worlds model of uncertain graphs, we define random walks on uncertain graphs for the first time and show that our definition of random walks satisfies Markov's property. We formulate the SimRank measure based on random walks on uncertain graphs. We discover a critical difference between random walks on uncertain graphs and random walks on deterministic graphs, which makes all existing SimRank computation algorithms on deterministic graphs inapplicable to uncertain graphs. To efficiently compute SimRank similarities, we propose three algorithms, namely the baseline algorithm with high accuracy, the sampling algorithm with high efficiency, and the two-phase algorithm with comparable efficiency as the sampling algorithm and about an order of magnitude smaller relative error than the sampling algorithm. The extensive experiments and case studies verify the effectiveness of our SimRank measure and the efficiency of our SimRank computation algorithms.
\end{abstract}

\section{Introduction}
\label{Sec:Intro}

Complicated relationships between entities are often represented by a graph. The similarities between entities can be revealed by analyzing the links between the vertices in a graph.
Recently, evaluating similarities between vertices has become a fundamental issue in graph analytics. It plays an important role in many applications, including entity resolution \cite{bhattacharya2006entity}, recommender systems \cite{fouss2007random} and spams detection \cite{benczur2006link}. Assessing similarities between vertices is also a cornerstone of many graph mining tasks, such as graph clustering \cite{zheng2013efficient}, frequent subgraph mining \cite{zhao2009p} and dense subgraph discovery \cite{zou2013structural}.

A lot of similarity measures have been proposed, e.g., Jaccard similarity \cite{jaccardsim},  Dice similarity \cite{dicesim} and cosine similarity \cite{cosinesim}, which are motivated by the intuition that two vertices are more similar if they share more common neighbors. However, these measures cannot evaluate similarities between vertices with no common neighbors. To address this problem, Jeh and Widom \cite{jeh2002simrank} proposed a versatile similarity measure called {\em SimRank} based on the intuition that {\em two vertices are similar if their in-neighbors are similar too}. Since SimRank captures the topology of the whole graph, it can be used to assess the similarity between two vertices regardless if they have common neighbors. Hence, SimRank has been widely used over the last decade. A lot of studies \cite{ fogaras2005scaling,fujiwara2013efficient,kusumoto2014scalable,lee2012top,li2010fast, lizorkin2008accuracy, shao2015efficient,tao2014efficient, yu2014fast, yu2015efficient} have been done on efficient SimRank computations.

Almost all the studies on SimRank focus on deterministic graphs. However, in recent years, people have realized that uncertainty is intrinsic in graph structures, e.g., protein-protein interaction (PPI) networks. A graph inherently accompanied with uncertainty is called an {\em uncertain graph}. Considerable researches on managing and mining uncertain graphs \cite{du2015probabilistic,jin2011discovering,kollios2013clustering, potamias2010kNN, zou2010finding,zou2013structural} have shown that the effects of uncertainty on the quality of  results have been undervalued in the past. To the best of our knowledge, the only work of SimRank computation on uncertain graphs has been carried out by Du et al.~\cite{du2015probabilistic}. Whereas, SimRank on uncertain graphs is important in many applications. We show two examples as follows.

\noindent\underline{\bf Application 1 (Detecting Similar Proteins).} Finding proteins with similar biological functions is of great significance in biology and pharmacy \cite{whalen2015sequence,nelson2007novel}. Traditionally, the similarity between proteins are measured by matching their corresponding DNA's \cite{whalen2015sequence}. However, similar DNA sequences may not generate proteins with similar functions. Recent approaches are based on protein-protein interaction (PPI) networks. A PPI network represents interactions between proteins detected by high-throughput experiments. A PPI network reflects functional relationships among proteins more directly. A pair of proteins with high structural-context similarity in a PPI network are more likely to have similar biological functions. However, due to errors and noise in high-throughput biological experiments, uncertainty is inherent in a PPI network. This motivates us to sutdy SimRank similarities on uncertain graphs.

\noindent{\underline{\bf Application 2 (Entity Resolution).}} Entity Resolution (ER) is a primitive operation in data cleaning. The goal of ER is to find records that refer to the same real-world entity from heterogeneous data sources \cite{bhattacharya2006entity}. Considerable ER algorithms \cite{bhattacharya2006entity,li2010eif,yin2007object} fall into the category of organizing data records as a graph, where vertices represent data records, and edges between records are associated with similarity values. Such graph is typically an uncertain graph since the weights are often normalized into $[0,1]$ and  regarded as probabilities. In the existing graph-based ER algorithms, they aggregate similar vertices into an entity but ignore uncertainty information. For example, the {\sf EIF} algorithm \cite{li2010eif} discards the edges whose weights are less than a threshold and aggregates similar records according to the Jaccard similarity. To take uncertainty into account, we study SimRank similarities on uncertain graphs.

\noindent{\underline{\bf Challenges.}} The challenges of SimRank on uncertain graphs come from two aspects, namely its formulation and computation. Notice that the SimRank on a deterministic graph can be formulated in the language of {\em random walks on graphs} \cite{jeh2002simrank}. Specifically, the SimRank matrix (i.e., the matrix of SimRank similarities between all pairs of vertices) can be formulated as a (nonlinear) combination of the one-step transition probability matrix (i.e., the matrix of one-step transition probabilities between all pairs of vertices). However, such formulation cannot be adapted to uncertain graphs. This is because, for a deterministic graph, the $k$-step transition probability matrix $\W^{(k)}$ equals the $k$th power of the one-step transition probability matrix $\W^{(1)}$, that is, $\W^{(k)} = (\W^{(1)})^k$. However, as analyzed in this paper, for an uncertain graph, the $k$-step transition probability matrix $\UW^{(k)}$ is \emph{unequal} to the $k$th power of the one-step transition probability matrix $\UW^{(1)}$, i.e., $\UW^{(k)} \ne (\UW^{(1)})^k$. Unfortunately, the only work of SimRank on uncertain graphs \cite{du2015probabilistic} does not solve this problem because it makes an unreasonable assumption that $\UW^{(k)} = (\UW^{(1)})^k$ for all $k \ge 1$. Hence, the first two challenges are as follows.
\begin{itemize}\setlength{\itemsep}{0pt}
\item[\bf 1:] How to define random walks on uncertain graphs?
\item[\bf 2:] How to define SimRank on uncertain graphs based on random walks on uncertain graphs?
\end{itemize}

For a deterministic graph, the SimRank matrix can be approximated using many methods \cite{lizorkin2008accuracy, tao2014efficient, yu2013towards, yu2014fast, yu2015efficient,yu2012space}. All these methods are based on the fact that the SimRank matrix is a (nonlinear) combination of the one-step transition probability matrix $\W^{(1)}$. Therefore, the central operations involved in these methods are matrix multiplications with the columns of $\W^{(1)}$ and $(\W^{(1)})^T$. However, all these methods cannot be adapted to compute the SimRank matrix for an uncertain graph because the $k$-step transition probability matrix $\UW^{(k)}$ on uncertain graphs does not satisfy that $\UW^{(k)} = (\UW^{(1)})^k$. In fact, the SimRank matrix for an uncertain graph is a combination of all transition probability matrices $\UW^{(k)}$ for $k = 1, 2, \ldots$. Hence, the challenges in computations are as follows.
\begin{itemize}\setlength{\itemsep}{0pt}
 \item[\bf 3:] How to efficiently compute the $k$-step transition probability matrix $\UW^{(k)}$ for an uncertain graph?
\item [\bf 4:] How to efficiently approximate the SimRank matrix for an uncertain graph?
\end{itemize}

To deal with the challenges C1--C4 listed above, we study the theory and algorithms on SimRank on uncertain graphs. The studies in this paper are strictly based on the {\em possible world model} of uncertain graphs \cite{jin2011discovering, kollios2013clustering, potamias2010kNN, zou2009frequent}. In the possible world model, an uncertain graph represents a probability distribution over all the possible worlds of the uncertain graph. Each possible world is a deterministic graph that the uncertain graph could possibly be in practice. The main contributions of this paper are as follows.

\noindent{\underline{\bf Contribution 1.}} To the best of knowledge, we are the first to formulate random walks on uncertain graphs totally following the possible world model. We define the $k$-step transition probability from a vertex $u$ to a vertex $v$ as the probability that a walk stays at $u$ at time $n$ and arrives at $v$ at time $n + k$ in a randomly selected possible world. Our definition satisfies {\em Markov's property}, that is, for all $n \ge 0$ and all vertices $v$, the probability that a walk stays at $v$ at time $n + 1$ is only determined by the vertex at which the walk stays at time $n$, independent of all the vertices that the walk has visited at time $0, 1, 2, \dots, n - 1$. One of our main findings is that, for an uncertain graph, the $k$-step transition probability matrix $\UW^{(k)}$ is {\em not equal to} the $k$th power of the one-step transition probability matrix $\UW^{(1)}$. In case when there is no uncertainty involved in graphs, our definition of random walks on uncertain graphs degenerates to random walks on deterministic graphs.

\noindent{\underline{\bf Contribution 2.}} Based on the model of random walks on uncertain graphs, we define the SimRank measure on uncertain graphs. The SimRank similarity between two vertices $u$ and $v$ is formulated as the combination of the probabilities that two random walks starting from $u$ and $v$, respectively, meet at the same vertex after $k$ transitions for all $k = 1, 2, \ldots$. Since $\UW^{(k)} \ne (\UW^{(1)})^k$, we cannot formulate the SimRank matrix in a recursive form of $\UW^{(1)}$ only. Thus, the existing algorithms for SimRank computations cannot be used to evaluate SimRank similarities on uncertain graphs.

\noindent{\underline{\bf Contribution 3.}} We propose three algorithms for approximating the SimRank similarity between two vertices. The central idea of these algorithms is approximating the SimRank similarity between two vertices $u$ and $v$ by combining the probabilities that two random walks starting from $u$ and $v$, respectively, meet at the same vertex after $k$ transitions for $1 \le k \le n$, where $n$ is a sufficiently large number. We prove that the approximate value converges to the exact value as $n \to +\infty$. Moreover, the approximation error exponentially decreases as $n$ becomes larger.

The three SimRank computation algorithms proposed in this paper adopt different approaches to computing transition probability matrices. The first algorithm exactly computes the transition probability matrices $\UW^{(1)}, \UW^{(2)}, \ldots, \UW^{(n)}$. The second algorithm approximates $\UW^{(1)}, \UW^{(2)}, \ldots, \UW^{(n)}$ via sampling. To make a tradeoff between efficiency and accuracy, we propose the third algorithm called the two-phase algorithm, which works in two phases. Let $1 \le l \le n$. In the first phase, we exactly compute $\UW^{(1)}, \UW^{(2)}, \ldots, \UW^{(l)}$; in the second phase, we approximate $\UW^{(l + 1)}, \UW^{(l + 2)}, \ldots, \UW^{(n)}$ by sampling. Finally, we combine these results to approximate the SimRank similarities. By carefully selecting $l$, the two-phase algorithm can achieve comparable efficiency as the sampling algorithm and about an order of magnitude smaller relative error than the sampling algorithm. Furthermore, we develop a new technique to share the common steps within a large number of independent sampling processes, which decreases the total sampling time by $1$--$2$ orders of magnitude.

\noindent{\underline{\bf Contribution 4.}} We conducted extensive experiments on a variety of uncertain graph datasets to evaluate our proposed algorithms. The experimental results verify both the effectiveness and the convergence of our SimRank measure. The two-stage algorithm is much more efficient than the baseline algorithm on large uncertain graphs, and its relative error is about an order of magnitude smaller than the sampling algorithm. Moreover, our speeding-up technique can make the sampling process $1$--$2$ orders of magnitude faster without harming the relative errors of the results. We also performed two interesting case studies on detecting similar proteins and entity resolution as we stated above. The results verify the effectiveness of our SimRank similarity measure.

The rest of this paper is organized as follows. Section 2 reviews some preliminaries. Section 3 gives a formal definition of random walks on uncertain graphs. Section 4 proposes the algorithm for computing the $k$-step transition probability matrices of an uncertain graph. Section 5 formulates the SimRank measure on uncertain graphs. Section 6 proposes three SimRank computation algorithms and the speeding-up technique. Section 7 reports the extensive experimental results. Section 8 overviews the related work. Finally, this paper is concluded in Section 9.

\section{Preliminaries}
\label{Sec:Preliminaries}

In this section we review some preliminary knowledge, including {\em random walks on graphs}, the {\em SimRank similarity measure} and the model of {\em uncertain graphs}.

\noindent\underline{\bf Random Walks on Graphs.}
A (deterministic) {\em graph} is a pair $(V, E)$, where $V$ is a set of {\em vertices}, and $E \subseteq V \times V$. Each element $(u, v) \in E$ is said to be an {\em arc} connecting vertex $u$ to vertex $v$. In this paper we consider {\em directed} graphs, in which $(u, v)$ and $(v, u)$ refer to different arcs.

Let $G$ be a directed graph. We use $V(G)$ and $E(G)$ to denote the vertex set and the arc set of $G$, respectively. A vertex $u$ is said to be an {\em in-neighbor} of a vertex $v$ if $(u, v)$ is an arc. Meanwhile, $v$ is an {\em out-neighbor} of $u$. Let $I_G(v)$ and $O_G(v)$ denote the sets of in-neighbors and out-neighbors of a vertex $v$ in $G$, respectively. A {\em walk} on $G$ is a sequence of vertices $W = v_0, v_1, \ldots, v_n$ such that $(v_i, v_{i + 1})$ is an arc for $0 \le i \le n - 1$. The {\em length} of $W$, denoted by $|W|$, is $n$. A sequence of random variables $X_0, X_1, X_2, \ldots$ over $V(G)$ is a {\em random walk} on $G$ if it satisfies {\em Markov's property}, that is,
\begin{equation*}
	\begin{split}
		~ & \Pr(X_i = v_i | X_0 = v_0, X_1 = v_1, \ldots, X_{i - 1} = v_{i - 1}) \\
		= & \Pr(X_i = v_i | X_{i - 1} = v_{i - 1})
	\end{split}
\end{equation*}
for all $i \ge 1$ and all $v_0, v_1, \ldots, v_i \in V(G)$. For any $u, v \in V(G)$, $\Pr(X_i = v | X_{i - 1} = u)$ represents the probability that the random walk, when on vertex $u$ at time $i - 1$, will next make a transition onto vertex $v$ at time $i$. Particularly,
\begin{equation*}\label{Eqn:TransProb-1}
	\Pr(X_i = v | X_{i - 1} = u) =
	\begin{cases}
		\frac{1}{|O_G(u)|} & \text{if } (u, v) \in E(G),\\
		0 & \text{otherwise}.
	\end{cases}
\end{equation*}
Note that $\Pr(X_i = v | X_{i - 1} = u)$ is fixed for all $i \ge 1$, so we denote the value of $\Pr(X_i = v | X_{i - 1} = u)$ by $\Pr(u \to_1 v)$, which is called the {\em one-step transition probability} from vertex $u$ to vertex $v$. Therefore, for all $v_0, v_1, \ldots, v_n \in V(G)$, the probability that a random walk $X_0, X_1, X_2, \ldots$, when starting from vertex $v_0$ at time $0$, will later be on vertex $v_i$ at time $i$ for $1 \leq i \leq n$ is
\begin{equation}\label{Eqn:ProbRW}
	\begin{split}
		~ & \Pr(X_1 = v_1, X_2 = v_2, \ldots, X_n = v_n | X_0 = v_0) \\
		= & \prod_{i = 1}^n \Pr(v_{i - 1} \to_1 v_i).
	\end{split}
\end{equation}

For all $i \ge 0$, $k \ge 1$, and all $u, v \in V(G)$, the probability that a random walk on vertex $u$ at time $i$ will later be on vertex $v$ after $k$ additional transitions is
\begin{equation*}\label{Eqn:TransProb}
	\begin{split}
		~ & \Pr(X_{i + k} = v | X_i = u) \\
		= & \sum_{\mathclap{\qquad \qquad v_1, v_2, \ldots, v_{k - 1} \in V(G)}}\Pr(X_{i + 1} = v_1, \ldots, X_{i + k - 1} = v_{k - 1}, X_{i + k} = v | X_i = u),
	\end{split}
\end{equation*}
By Eq.~\eqref{Eqn:ProbRW}, $\Pr(X_{i + k} = v | X_i = u)$ is fixed for all $i \ge 0$, so we denote the value of $\Pr(X_{i + k} = v | X_i = u)$ by $\Pr(u \to_k v)$, which is called the {\em $k$-step transition probability} from vertex $u$ to vertex $v$. For all $k \ge 1$, $\Pr(u \to_k v)$ can be recursively formulated by
\begin{equation*}
	\small
	\Pr(u \to_k v) = \sum_{w \in I_G(v)} \Pr(u \to_{k - 1} w) \Pr(w \to_1 v).
\end{equation*}

We can also formulate transition probabilities in the form of matrices. Suppose $V(G) = \{v_1, v_2, \ldots, v_n\}$. For $k \ge 1$, let $\W^{(k)}$ be the matrix of $k$-step transition probabilities, that is, $\W^{(k)}_{i,j} = \Pr(v_i \to_k v_j)$ for $1 \le i, j \le n$. We have $\W^{(1)} = \A$, where $\A$ is the {\em adjacency matrix} of $G$ with rows normalized, that is, $\A_{i,j} = 1/|O_G(v_i)|$ if $(v_i, v_j) \in E(G)$, and $\A_{i,j} = 0$ otherwise. For $k > 1$, we have
\begin{equation*}\label{Eqn:DWA}
	\W^{(k)} = \W^{(k - 1)} \W^{(1)} = \A^k.
\end{equation*}

\noindent\underline{\bf SimRank.}
{\em SimRank} is a structural-context similarity measure for vertices in a directed graph \cite{jeh2002simrank}. It is designed based on the intuition that two vertices are similar if their in-neighbors are similar too. Let $s(u, v)$ be the {\em SimRank similarity} between vertices $u$ and $v$ in a directed graph $G$. $s(u, v)$ is defined by
\begin{equation}\label{eqn:SimRankBasic}
    s(u,v) =
    \begin{cases}
    	1 & \text{if } u = v, \\
    	\frac{c\sum_{u' \in I_{G}(u)} \sum_{v' \in I_{G}(v)} s(u',v')}{|I_{G}(u)||I_{G}(v)|} & \text{otherwise},
    \end{cases}	
\end{equation}
where $0 < c < 1$ is called the {\em delay factor}. The system of Eq.~\eqref{eqn:SimRankBasic} can be reformulated using a matrix equation. Let $V(G) = \{v_1, v_2, \ldots, v_n\}$ and $\Sim$ be a matrix with $n$ rows and $n$ columns, where $\Sim_{i,j} = s(v_i, v_j)$ for $1 \le i, j \le n$. Let $\A$ be the column-normalized adjacency matrix of graph $G$. We have
\begin{equation*}
	\Sim = c \A^T \Sim \A - \mathrm{diag}(c \A^T \Sim \A) + \I,
\end{equation*}
where $\mathrm{diag}(\X)$ is the diagonal matrix whose diagonal components are the diagonal components of $\X$. In many literatures \cite{du2015probabilistic,kusumoto2014scalable, li2010fast,yu2015efficient}, $\Sim$ is often approximated as
\smallskip
\begin{equation} \label{Eqn:simrankd}
	\small
	\Sim = c \A^T \Sim \A + (1 - c)\I.
\end{equation}

\begin{figure}[t]
	\centering
	\subfigure[Uncertain graph.]{\includegraphics[width=0.4\columnwidth]{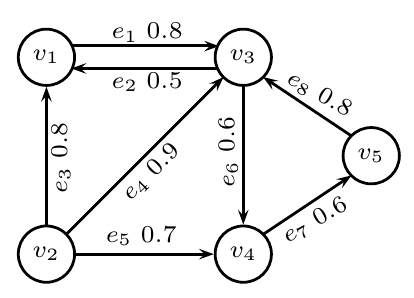} \label{Fig:UGraph:UG}}
	\subfigure[Possible world.]{\includegraphics[width=0.4\columnwidth]{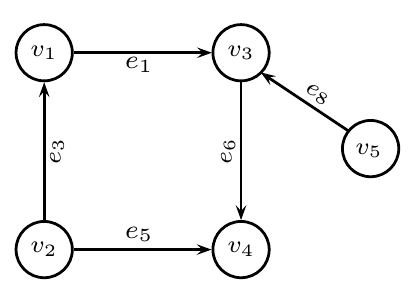} \label{Fig:UGraph:PW}}
	\caption{An uncertain graph and one of its possible worlds.}
	\label{Fig:UGraph}
\end{figure}

\noindent\underline{\bf Uncertain Graphs.}
An \emph{uncertain graph} is a tuple $(V,E,P)$, where $V$ is a set of vertices, $E$ is a set of arcs, and $P: E \rightarrow (0,1]$ is a function assigning {\em existence probabilities} to the arcs. Particularly, $P(e)$ is the probability that arc $e$ exists in practice. For clarity, we denote an uncertain graph by a written letter such as $\G$ and denote a deterministic graph by a printed letter such as $G$. Let $V(\G)$, $E(\G)$ and $P_{\G}$ be the vertex set, the edge set and the existence probability function of an uncertain graph $\G$, respectively. Let $I_{\G}(v)$ and $O_{\G}(v)$ be the sets of in-neighbors and out-neighbors of vertex $v$ in an uncertain graph $\G$, respectively.

Under the \emph{possible world semantics} \cite{jin2011discovering,kollios2013clustering,potamias2010kNN,zou2009frequent,zou2013structural}, an uncertain graph $\G$ represents a probability distribution over all its \emph{possible worlds}. More precisely, a \emph{possible world} of $\G$ is a deterministic graph $G$ such that $V(G) = V(\G)$ and $E(G) \subseteq E(\G)$. Let $\Omega(\G)$ be the set of all possible worlds of $\G$ and $\G \Rightarrow G$ be the event that $\G$ exists in the form of its possible world $G$ in practice. Following previous works \cite{jin2011discovering,kollios2013clustering,potamias2010kNN,zou2009frequent,zou2010finding,zou2013structural,du2015probabilistic}, we reasonably assume that the existence probabilities of edges are mutually independent. Hence, the probability of event $\G \Rightarrow G$ is
\begin{equation}\label{eqn:uncertaingraph}
	\small
	\Pr(\G \Rightarrow G) = \prod_{e \in E(G)} P_{\G}(e) \cdot \prod_{e \in E(\G) \setminus E(G)} (1-P_{\G}(e)).
\end{equation}
It is easy to verify that $\sum_{G \in \Omega(\G)} \Pr(\G \Rightarrow G) = 1$. Fig.~\ref{Fig:UGraph} shows an uncertain graph $\G$ and one of its possible worlds $G$. For the possible world $G$ in Fig.~\ref{Fig:UGraph:PW}, we have $\Pr(\G \Rightarrow G) = P_{\G}(e_1)P_{\G}(e_3)P_{\G}(e_5) P_{\G}(e_6)
P_{\G}(e_8)(1-P_{\G}(e_2))(1-P_{\G}(e_4))(1-P_{\G}(e_7)) \approx 0.0043$.

\section{Random Walks on Uncertain Graphs} \label{sec:RWUG}
\label{Sec:UncertainRandomWalks}

In this section we give a formal definition of a random walk on an uncertain graph. Let $\G$ be an uncertain graph. Under the possible world model, $\G$ encodes a probability distribution over $\Omega(\G)$, the set of all possible worlds of $\G$. Let $X_0, X_1, X_2, \ldots$ be a random walk on any possible world of $\G$. For all $v_0, v_1, \ldots, v_n \in V(\G)$, the probability $\Pr(X_n = v_n | X_0 = v_0, \ldots, X_{n - 1} = v_{n - 1})$ generally takes different values on different possible worlds of $\G$.

Let $\Pr_G(E)$ denote the probability of an event $E$ on a possible world $G$, and let $\Pr_{\G}(E)$ denote the probability of an event $E$ on a possible world of $\G$ selected at random according to the probability distribution given in Eq.~\eqref{eqn:uncertaingraph}. Then, we have

\begin{equation}\label{Eqn:UncertainMarkovProperty}
	\small
	\begin{split}
		~ & \Pr\nolimits_{\G}(X_n = v_n | X_0 = v_0, \ldots, X_{n - 1} = v_{n - 1}) \\
		= & \sum_{\mathclap{G \in \Omega(\G)}} \Pr\nolimits_G(X_n = v_n | X_0 = v_0, \ldots, X_{n - 1} = v_{n - 1}) \Pr(\G \Rightarrow G) \\
		= & \sum_{\mathclap{G \in \Omega(\G)}} \Pr\nolimits_G(X_n = v_n | X_{n - 1} = v_{n - 1}) \Pr(\G \Rightarrow G) \\
		= & \Pr\nolimits_{\G}(X_n = v_n | X_{n - 1} = v_{n - 1}).
	\end{split}
\end{equation}
The second equality is due to Markov's property of a random walk on a deterministic graph. The above equation states that, on an uncertain graph, the probability that a random walk is on vertex $v_n$ at time $n$ is independent of all the previous vertices except the vertex it stays at time $n - 1$.

We now define the {\em $k$-step transition probability} from a vertex $u$ to a vertex $v$ on  uncertain graph $\G$ for $k \ge 1$. On a randomly chosen possible world of $\G$, the probability that a random walk is on vertex $v$ at time $n$ given that it is on vertex $u$ at time $n - k$ is given by
\begin{equation}\label{Eqn:TransProbUG}
	\small
	\begin{split}
		~ & \Pr\nolimits_{\G}(X_n = v | X_{n - k} = u) \\
		= & \sum_{G \in \Omega(\G)} \Pr\nolimits_G(X_n = v | X_{n - k} = u) \Pr(\G \Rightarrow G) \\
		= & \sum_{G \in \Omega(\G)} \Pr\nolimits_G(u \to_k v) \Pr(\G \Rightarrow G).
	\end{split}
\end{equation}
Of course, $\Pr_{\G}(X_n = v | X_{n - k} = u)$ is fixed regardless of $n$, so we use $\Pr_\G(u \to_k v)$ to denote the value of $\Pr\nolimits_{\G}(X_n = v | X_{n - k} = u)$ for any $n \ge k$.

For $k \ge 1$, let $\UW^{(k)}$ be the matrix of $k$-step transition probabilities on uncertain graph $\G$, where $\UW^{(k)}_{i, j} = \Pr_\G(v_i \to_k v_j)$. By Eq.~\eqref{Eqn:TransProbUG}, we have
\begin{equation*} \label{Eqn:TransUW}
	\small
	\UW^{(k)} = \sum_{G \in \Omega(\G)} \Pr(\G \Rightarrow G) \W^{(k)}_G,
\end{equation*}
where $\W^{(k)}_G$ is the matrix of $k$-step transition probabilities on possible world $G$.

\section{Computing Transition Probabilities}
\label{sec:RWCom}

In this section we propose an algorithm for computing the $k$-step transition probability, $\Pr_{\G}(u \rightarrow_{k} v)$, from a vertex $u$ to a vertex $v$ in an uncertain graph $\G$. By Eq.~\eqref{Eqn:TransProbUG}, we have
\begin{equation}\label{Eqn:TransExd}
	\small
	\begin{split}
		~ & \Pr\nolimits_{\G}(u \to_k v) = \Pr\nolimits_{\G}(X_k = v | X_{0} = u)\\
		= & \sum_{\mathclap{\qquad \qquad v_1, v_2, \dots, v_{k-1} \in V(\G)}} \Pr\nolimits_{\G}(X_1 = v_1, X_2 = v_2, \dots, X_k =v |X_0 = u).
	\end{split}
\end{equation}
In the above equation, if $u, v_1, v_2, \ldots, v_{k - 1}, v$ is not a walk, then $\Pr_{\G}(X_1 = v_1, X_2 = v_2, \dots, X_{k - 1} = v_{k - 1}, X_k = v |X_0 = u) = 0$. Hence, we only need to consider the vertices $v_1, v_2, \ldots, v_{k - 1}$ such that $u, v_1, v_2, \ldots, v_{k - 1}, v$ is a walk.

For convenience of presentation, let $v_0 = u$ and $v_k = v$. For any walk $W = v_0,v_1, \dots, v_k$, we call $\Pr_{\G}(X_1 = v_1, X_2 = v_2, \dots, X_k =v_k |X_0 = v_0)$ the {\em walk probability} of $W$ given that $W$ starts from $v_0$. According to the possible world model,
\begin{equation*} \label{Eqn:RWPro}
	\scriptsize
	\begin{split}
		~ & \Pr\nolimits_{\G}(X_1 = v_1, X_2 = v_2, \dots, X_k =v_k |X_0 = v_0) \\
		= & \sum_{G \in \Omega(\G)} \Pr\nolimits_G(X_1 = v_1, X_2 = v_2, \dots, X_k =v_k |X_0 = v_0) \Pr(\G \Rightarrow G)\\
    	= & \sum_{G \in \Omega(\G)} \,  \prod_{i = 0}^{k-1} \Pr\nolimits_G(X_{i+1} = v_{i+1} | X_i = v_i) \Pr(\G \Rightarrow G).
	\end{split}
\end{equation*}
Hence, computing the $k$-step transition probability $\Pr_{\G}(u \rightarrow_{k} v)$ reduces to computing the walk probabilities of all walks starting from $u$ and staying at $v$ after $k$ additional transitions.

Note that, on a deterministic graph $G$, the walk probability $\Pr_G(X_1 = v_1, X_2 = v_2, \dots, X_k = v_k |X_0 = v_0)$ can be easily computed by Eq.~\eqref{Eqn:ProbRW}, that is,
\begin{equation*}
	\small
	\Pr\nolimits_G(X_1 = v_1, \ldots, X_k = v_k | X_0 = v_0) = \prod_{i = 1}^k \Pr\nolimits_G(v_{i - 1} \to_1 v_i).
\end{equation*}
However, this simple method cannot be generalized to computing walk probabilities on an uncertain graph because
\begin{equation*} \label{Eqn:RWProWrg1}
	\small
	\begin{split}
		~ & \Pr\nolimits_{\G}(X_1 = v_1, X_2 = v_2, \dots, X_k = v_k |X_0 = v_0)\\
		= & \sum_{G \in \Omega(\G)} \left( \Pr(\G \Rightarrow G) \prod_{i = 0}^{k-1} \Pr\nolimits_G(v_i \to_1 v_{i + 1}) \right)
  \end{split}
\end{equation*}
and
\begin{equation*} \label{Eqn:RWProWrg2}
	\small
	\prod_{i = 0}^{k-1} \Pr\nolimits_{\G}(v_i \to_1 v_{i + 1})
		= \prod_{i = 0}^{k-1} \sum_{G \in \Omega(\G)} \Pr(\G \Rightarrow G) \Pr\nolimits_G(v_i \to_1 v_{i + 1}).
\end{equation*}
The two equations above are generally \emph{unequal} unless none of $v_0, v_1, \dots, v_k$ are the same. Intuitively, if $v_i = v_j$ for some $1 \le i, j \le k$ and $i \ne j$, then, on any possible world $G$, the transition from $v_i$ to $v_{i + 1}$ and the transition from $v_j$ to $v_{j + 1}$ are not independent. This finding distinguishes our work from the work by Du et al.~\cite{du2015probabilistic}, in which they make an unreasonable assumption that $\Pr\nolimits_{\G}(X_1 = v_1, \ldots, X_k = v_k | X_0 = v_0) = \prod_{i = 1}^k \Pr\nolimits_{\G}(v_{i - 1} \to_1 v_i)$.

In the following we propose an algorithm for computing walk probabilities in Section \ref{sec:comwp} and an algorithm for computing $k$-step transition probabilities in Section \ref{sec:comtp}.

\subsection{Computing Walk Probabilities}
\label{sec:comwp}

Let $W = v_0,v_1, \dots, v_k$ be a walk on uncertain graph $\G$. Let $V(W)$ be the set of vertices in $W$. We have $|V(W)| \le k + 1$ because a vertex may appear multiple times in $W$. For every vertex $v \in V(W)$, let $O_W(v)$ be the set of out-neighbors of $v$ in $W$, and let $c_W(v)$ be the number of occurrences of arcs from $v$ to a vertex in $W$. We have $c_W(v) \ge |O_W(v)|$ because the walk may transit from $v$ to a certain vertex in $O_W(v)$ multiple times.

For all possible worlds $G$ of $\G$, if $W$ is a walk in $G$, it follows from Eq.~\eqref{Eqn:ProbRW} that
\begin{equation*}\label{Eqn:produ}
	\begin{split}
		~ & \Pr\nolimits_{G}(X_1 = v_1, X_2 = v_2, \ldots, X_k = v_k | X_0 = v_0) \\
		= & \prod_{v \in V(W)} \mathrm{inv}(|O_G(v)|)^{c_W(v)},
	\end{split}
\end{equation*}
where $\mathrm{inv}(x) = 1 / x$ if $x \ne 0$, and $\mathrm{inv}(x) = 1$ otherwise. If $W$ is not a walk in $G$, then $\Pr_{G}(X_1 = v_1, X_2 = v_2, \ldots, X_k = v_k | X_0 = v_0) = 0$. Therefore, the walk probability $\Pr\nolimits_{\G}(X_1 = v_1, X_2 = v_2, \ldots, X_k = v_k | X_0 = v_0)$ can be computed by
\begin{equation}\label{Eqn:prodUGu}
	\small
	\begin{split}
		~ & \Pr\nolimits_{\G}(X_1 = v_1, X_2 = v_2, \ldots, X_k = v_k | X_0 = v_0) \\
		= & \sum_{G \in \Omega(\G), W \vdash G} \Pr(\G \Rightarrow G) \prod_{v \in V(W)} \mathrm{inv}(|O_G(v)|)^{c_W(v)},
	\end{split}
\end{equation}
where $W \vdash G$ represents that $W$ is a walk in $G$.

Due to the independence assumption on the edges of $\G$, we have the following lemma, which gives an equivalent formulation of Eq.~\eqref{Eqn:prodUGu}.

\begin{lemma}\label{Lem:prodUGTrans}
\begin{equation}\label{Eqn:prodUGTrans}
	\small
	\Pr\nolimits_{\G}(X_1 = v_1, \ldots, X_k = v_k | X_0 = v_0) = \prod_{v \in V(W)} \alpha_W(v),
\end{equation}
where
\begin{equation*}
	\small
	\alpha_W(v) = \sum_{G} \Pr(\G \Rightarrow G) \mathrm{inv}(|O_G(v)|)^{c_W(v)}.
\end{equation*}
The summation in the equation for $\alpha_W(v)$ is over all possible worlds $G$ such that $(v, w) \in E(G)$ for all $w \in O_W(v)$.
\end{lemma}

\begin{proof}
First, we rewrite Eq.~\eqref{Eqn:prodUGu} as
\begin{equation}\label{Eqn:WalkPExp}
	\small
	\begin{split}
		~ & \Pr\nolimits_{\G}(X_1 = v_1, X_2 = v_2, \ldots, X_k = v_k | X_0 = v_0) \\
		= & \E[\Pr\nolimits_{G}(X_1 = v_1, X_2 = v_2, \ldots, X_k = v_k | X_0 = v_0)] \\
        = & \E[\prod_{v \in V(W)} \mathrm{inv}(|O_G(v)|)^{c_W(v)}].
	\end{split}
\end{equation}
That is, the walk probability of $W$ on $\G$ is the expectation of the walk probability over all the possible world graphs $G$ of $\G$. Since the out edges of a vertex $v$ is independent of the out edges of another vertex $u$. Thus the expectation and the product operations in Eq.~\eqref{Eqn:WalkPExp} could be exchanged as

\begin{equation*}
	\small
	\begin{split}
		~ & \Pr\nolimits_{\G}(X_1 = v_1, X_2 = v_2, \ldots, X_k = v_k | X_0 = v_0) \\
        = & \E[\prod_{v \in V(W)} \mathrm{inv}(|O_G(v)|)^{c_W(v)}] \\
        = & \prod_{v \in V(W)} \E[\mathrm{inv}(|O_G(v)|)^{c_W(v)}] \\
        = & \prod_{v \in V(W)} (\sum_{G} \Pr(\G \Rightarrow G) \mathrm{inv}(|O_G(v)|)^{c_W(v)}).
	\end{split}
\end{equation*}
Let $\alpha_W(v) = \sum_{G} \Pr(\G \Rightarrow G) \mathrm{inv}(|O_G(v)|)^{c_W(v)}$. The lemma holds.
\end{proof}

For all $v \in V(W)$, we can compute the term $\alpha_W(v)$ in Eq.~\eqref{Eqn:prodUGTrans} in polynomial time. The method is described as follows. Observe that, for each $v \in V(W)$, $c_W(v)$ is a constant; while $|O_G(v)|$ varies on different possible worlds $G$. Our method is based on evaluating  the probability distribution of $|O_G(v)|$ across all possible worlds $G$ of $\G$ such that $(v, w) \in E(G)$ for all $w \in O_W(v)$.

Let $O_{\G}(v) \setminus O_W(v) = \{w_1, w_2, \ldots, w_n\}$. For $0 \leq j \leq i \leq n$, let $r(i, j)$ represent the probability that only $j$ vertices in $w_1,w_2, \dots w_i$ are connected to $v$ in a randomly selected possible world of $\G$. Then, we have
\begin{align*}
	\small
	r(0, 0) &= 1, & \\
	r(i, 0) &= r(i-1, 0)(1 - P_{\G}((v, w_i))) & \text{for } 1 \leq i \leq n, \\
	r(i, i) &= r(i-1, i-1)P_{\G}((v, w_i)) & \text{for } 1 \leq i \leq n, \\
	r(i, j) &= r(i-1, j-1)P_{\G}((v, w_i)) \\
	& \quad + r(i-1, j)(1 - P_{\G}((v, w_i))) & \text{for } 1 \leq j < i \leq n.
\end{align*}
Naturally, the probability that $|O_G(v)| = x$ in a randomly selected possible world $G$ such that $(v, w) \in E(G)$ for all $w \in O_W(v)$ equals $r(n, x - |O_W(v)|)$. Thus, we have
\begin{equation}\label{Eqn:pruw}
	\small
    \begin{split}
    	\alpha_W(v) &= \prod_{w \in O_W(v)} P_{\G}((v, w)) \sum_{x = 0}^n r(n,x) \mathrm{inv}(x + |O_W(v)|)^{c_W(v)}.
    \end{split}
\end{equation}

By Lemma~\ref{Lem:prodUGTrans} and Eq.~\eqref{Eqn:pruw}, we immediately have the \textsf{WalkPr} algorithm as described in Figure \ref{Fig:WalkPr} for computing the walk probability of a walk $W$ in an uncertain graph $\G$. For every vertex $v \in V(W)$, lines \ref{Ln:BeginDP}--\ref{Ln:EndDP} compute the values $r(i, j)$ for $0 \leq j \leq i \leq n$ in a bottom-up manner, which totally run in $O({(|O_{\G}(v)| - |O_W(v)|)}^2)$ time. Using the values $r(n, x)$ for $0 \le x \le n$, line \ref{Ln:Alpha} computes $\alpha_W(v)$ in $O(|O_{\G}(v)|)$ time, and line \ref{Ln:pAlpha} multiplies $p$ by $\alpha_W(v)$. Thus, the running time of \textsf{WalkPr} is $\sum_{v \in V(W)} (|O_{\G}(v)| - |O_W(v)|)^2$. Let $d$ be the average out-degree of the vertices in $\G$. The time complexity of \textsf{WalkPr} is therefore $O(|W| d^2)$.

\begin{figure}[!t]
    \centering
	\scriptsize
	\fbox{
		\parbox{\columnwidth}{
		\textbf{Algorithm} \textsf{WalkPr}$(\G, W)$
		\begin{algorithmic}[1]
		    \STATE $p \gets 1$
		    \FORALL{$v \in V(W)$}
		    	\STATE //let $O_{\G}(v) \setminus O_W(v) = \{w_1, w_2, \ldots, w_n\}$ \label{Ln:BeginDP}
		    	\STATE $r(0, 0) \gets 1$
	    	    \FOR{$i \gets 1$ \TO $n$}
    	    	    \STATE $r(i, 0) \gets r(i-1, 0) (1 - P_{\G}((v, w_i)))$
        	    	\STATE $r(i, i) \gets r(i-1, i-1) P_{\G}((v, w_i))$
            		\FOR{$j \gets 1$ \TO $i - 1$}
                		\STATE $r(i, j) \gets r(i-1, j-1) P_{\G}((v, w_i)) + r(i-1, j)(1 - P_{\G}(v, w_i)))$
            		\ENDFOR
        		\ENDFOR \label{Ln:EndDP}
        		\STATE $\alpha \gets \prod\limits_{w \in O_W(v)} P_{\G}((v, w)) \sum\limits_{x = 0}^n r(n,x) \mathrm{inv}(x + |O_W(v)|)^{c_W(v)}$ \label{Ln:Alpha}
        		\STATE $p \gets p \alpha$ \label{Ln:pAlpha}
    		\ENDFOR
    		\RETURN $p$
		\end{algorithmic}
		}
	}
	\caption{Algorithm \textsf{\small WalkPr}.}
	\label{Fig:WalkPr}
\end{figure}

\noindent{\bf Example} Consider the uncertain graph $\G$ illustrated in Fig.~\ref{Fig:UGraph:UG}. We demonstrate how to compute the walk probability of a walk $W = v_1, v_3, v_1, v_3, v_4, v_2, v_3, v_4, v_2$ in $\G$ by Algorithm \textsf{WalkPr}. We have $V(W) = \{v_1, v_2, v_3, v_4\}$. For each $v \in V(W)$, we illustrate $O_{W}(v)$, $c_W(v)$, $O_{\G}(v) \setminus O_{W}(v)$, $r(n, x)$ and $\alpha_W(v)$ in Table~\ref{Tab:ExampleWalk}. Consequently, the output of \textsf{WalkPr} is $0.64 \times 0.54 \times 0.0375 \times 0.385 = 0.0049896$.

\begin{table}
    \caption{An example of computing walk probabilities.}
    \scriptsize
    \centering
    \begin{tabular}{@{\extracolsep{-0.08in}}c|cccc}
        \hline
        $v$ & $v_1$ & $v_2$ & $v_3$ & $v_4$\\ \hline\hline
        $O_{W}(v)$ & $\{v_3\}$ & $\{v_3\}$ & $\{ v_1, v_4\}$ & $\{ v_2 \}$ \\ \hline
        $c_{W}(v)$ & $2$ & $1$ & $3$ & $2$ \\ \hline
        $O_{\G}(v) \setminus O_{W}(v)$ & $\emptyset$ & $\{v_1\}$ & $\emptyset$ & $\{v_5\}$ \\ \hline
        $r(n,x)$ & $r(0,0) = 1$ & \tabincell{c}{$r(1,0) = 0.2$ \\ $r(1,1) = 0.8$} & $r(0,0) = 1$ & \tabincell{c}{$r(1,0) = 0.4$ \\ $r(1,1) = 0.6$} \\ \hline
        $\alpha_{W}(v)$ & $0.64$ & $0.54$ & $0.0375$ & $0.385$ \\ \hline
    \end{tabular}
    \label{Tab:ExampleWalk}
\end{table}

\subsection{Computing $k$-step Transition Probabilities}
\label{sec:comtp}

We now propose the algorithm for computing the $k$-step transition probability, $\Pr_{\G}(u \rightarrow_k v)$, from a vertex $u$ to a vertex $v$ in an uncertain graph $\G$. By Eq.~\eqref{Eqn:TransExd}, $\Pr(u \rightarrow_k v)$ is the summation of walk probabilities of all walks starting from $u$ and ending at $v$ after $k$ transitions. To further improve efficiency, instead of computing $\Pr(u \rightarrow_k v)$ from scratch, we compute $\Pr(u \rightarrow_k v)$ based on the $(k - 1)$-step transition probabilities $\Pr(u \rightarrow_{k - 1} w)$ for all vertices $w$ such that $(w, v)$ is an arc. Our incremental method is based on the following lemmas.

\begin{lemma} \label{Lem:incrw1}
Let $W = v_0, v_1, \ldots, v_k$ be a walk on $\G$ and $(v_k, v_{k + 1}) \in E(\G)$. Then, $W' = v_0, v_1, \ldots, v_{k + 1}$ is also a walk on $\G$, and
\begin{equation*}
	\small
	\frac{\Pr\nolimits_{\G}(X_1 = v_1, \ldots, X_{k + 1} = v_{k + 1} | X_0 = v_0)}{\Pr\nolimits_{\G}(X_1 = v_1, \ldots, X_k = v_k | X_0 = v_0)} = \frac{\alpha_{W'}(v_k)}{\alpha_{W}(v_k)}.
\end{equation*}
\end{lemma}

\begin{proof}
By Lemma~\ref{Lem:prodUGTrans}, we have
\begin{equation*}
	\small
	\Pr\nolimits_{\G}(X_1 = v_1, \ldots, X_k = v_k | X_0 = v_0) = \prod_{v \in V(W)} \alpha_W(v)
\end{equation*}
and
\begin{equation*}
	\small
	\Pr\nolimits_{\G}(X_1 = v_1, \ldots, X_k = v_k | X_0 = v_0) = \prod_{v \in V(W')} \alpha_{W'}(v).
\end{equation*}
Since $W$ and $W'$ only differs on the last vertex, only $\alpha_W(v)$ changes. Thus, we have $\alpha_W(v) = \alpha_{W'}(v)$ for all vertices $v \in V(W)$ and $v \neq v_k$. Hence, this lemma holds.
\end{proof}

In some cases when $W$ is not too long, the above lemma can be simplified as the following one.

\begin{lemma} \label{Lem:incrw2}
Let $W = v_0, v_1, \ldots, v_k$ and $W' = v_0, v_1, \ldots, v_{k + 1}$ be two walks on $\G$, where $(v_k, v_{k + 1})$ is an arc of $\G$. If the minimum length of the cycles in $\G$ is at least $k$, then we have
\begin{equation*}
	\small
    \frac{\Pr\nolimits_{\G}(X_1 = v_1, \ldots, X_{k + 1} = v_{k + 1} | X_0 = v_0)}
    {\Pr\nolimits_{\G}(X_1 = v_1, \ldots, X_k = v_k | X_0 = v_0)} =
    \Pr\nolimits_{\G}(v_k \to_1 v_{k + 1}).
\end{equation*}
\end{lemma}

\begin{proof}
Since the minimum length of the cycles in $\G$ is at least $k$, the out arcs of vertices $v_1,v_2, \dots, v_k$ in $W$ is at most $1$. Thus, we have
\begin{equation*}
	\small
    \alpha_W(v_k) = \Pr\nolimits_{\G}(v_k \to_1 v_{k + 1}),
\end{equation*}
for $v_1,v_2, \dots, v_k$ in $W$.

By Eq.~\eqref{Eqn:prodUGTrans}, the lemma holds.
\end{proof}

Based on Lemmas~\ref{Lem:incrw1} and \ref{Lem:incrw2}, we propose the {\sf TransPr} algorithm to compute $k$-step transition probabilities as described in Figure \ref{Fig:TransPr}. The input of \textsf{TransPr} is an uncertain graph $\G$ and an integer $K$. The output of \textsf{TransPr} is the $k$-step transition probability matrices $\UW^{(1)}, \UW^{(2)}, \ldots, \UW^{(K)}$.

The algorithm first computes $\UW^{(1)}$ and keeps it in main memory (line~\ref{line:ct:uw1}). The space used to store $\UW^{(1)}$ is $O(|E(\G)|)$. Then, we write $\UW^{(1)}$ to the $1$-step transition probability matrix file on disk. Particularly, for every $1$-step walk $W = u, v$, we write to disk a tuple composed by the walk $W$,
the walk probability of $W$, and the value $\alpha_W(v)$. Then, we compute the length $\ell$ of the shortest cycle in $\G$ using the algorithm proposed in \cite{horton1987polynomial} (line~\ref{line:ct:ell}).

In the main loop (lines~\ref{line:ct:loopstart}--\ref{line:ct:loopend}), we compute $\UW^{(k+1)}$ based on $\UW^{(k)}$ for $1 \leq k \leq K - 1$. For each specific $k$, we scan the walk probability file of $k$-step walks on disk. For each tuple $(W, p, \alpha)$ in the file, $W$ is a walk of length $k$, $p$ is the walk probability of $W$, and $\alpha$ is the value $\alpha_W(v)$, where $v$ is the last vertex in $W$. For every vertex $w \in O_{\G}(v)$, we  append $w$ to the end of $W$ and thus obtain a new walk $W'$ of length $k + 1$. We compute the walk probability of $W'$ based on the walk probability of $W$ either by Lemma \ref{Lem:incrw2} (line \ref{line:walkpr1}) or by Lemma \ref{Lem:incrw1} (line \ref{line:walkpr2}). Moreover, we compute the value $\alpha_{W'}(w)$. After this, we write $W'$, the walk probability of $W'$ and the value $\alpha_{W'}(w)$ to the walk probability file of $(k + 1)$-step walks on disk (line \ref{line:ct:k1}). When the scanning over the walk probability file of $k$-step walks is completed, we sort the tuples $(W', p', \alpha')$ in the walk probability file of $(k + 1)$-step walks according to the start and the end vertices of $W'$. For all walks with the same start vertex $u$ and the same end vertex $v$, we compute the $(k + 1)$-step transition probability $\Pr_{\G}(u \to_{k + 1} v)$ by summing up the walk probabilities of all these walks (line \ref{line:ct:merge}). Finally, we write the tuple $(u, v, \Pr_{\G}(u \to_{k + 1} v))$ to the file of the $(k + 1)$-step transition probability matrix $\UW^{(k + 1)}$ (line \ref{line:ct:output}).

\begin{figure}[!t]
    \scriptsize
    \fbox{
    \parbox{\columnwidth}{
    \textbf{Algorithm} \textsf{TransPr} \\
    \textbf{Input:} an uncertain graph $\G$ and an integer $K$\\
    \textbf{Output:} $\UW^{(1)}, \UW^{(2)}, \ldots, \UW^{(K)}$
    \begin{algorithmic}[1]
    	\STATE compute $\UW^{(1)}$ and write it to disk \label{line:ct:uw1}
        \STATE $\ell \gets$ the length of the shortest cycle in $\G$ \label{line:ct:ell}
        \FOR{$k \gets 1$ \TO $K - 1$} \label{line:ct:loopstart}
            \FOR{every tuple $(W, p, \alpha)$ in the walk probability file of $k$-step walks}
                \STATE $v \gets$ the last vertex in $W$
                \FORALL{$w \in O_{\G}(v)$}
                    \STATE $W' \gets$ the walk obtained by appending $v$ to $W$
                    \IF{$k \leq \ell$} \label{line:ct:incs}
                        \STATE $p' \gets p \UW^{(1)}_{vw}$ \label{line:walkpr1}
                        \STATE $\alpha' \gets 1$
                    \ELSE
                    	\STATE $\alpha' \gets \alpha_{W'}(w)$ \label{line:walkpr2}
                    	\STATE $p' \gets p \alpha' / \alpha$
                    \ENDIF \label{line:ct:ince}
                    \STATE write $(W',p',\alpha')$ to the walk probability file of $(k+1)$-step walks \label{line:ct:k1}
                \ENDFOR
            \ENDFOR
            \STATE sort all tuples $(W, p, \alpha)$ according to the start and the end vertices of $W$
            \FOR{each pair of vertices $u$ and $v$}
                \STATE $\Pr_{\G}(u \to_{k + 1} v) \gets$ summation of walk probabilities of all walks starting at $u$ and ending at $v$ \label{line:ct:merge}
                \STATE write $(u, v, \Pr_{\G}(u \to_{k + 1} v))$ to the file of $\UW^{(k + 1)}$ \label{line:ct:output}
            \ENDFOR
        \ENDFOR \label{line:ct:loopend}
    \end{algorithmic}
    }}
    \caption{Algorithm \textsf{\small TransPr}.}
    \label{Fig:TransPr}
\end{figure}

\section{SimRank Similarities on Uncertain Graphs}

In this section we give a formal definition of SimRank similarity on an uncertain graph. First, let us review a random-walk-based definition of SimRank similarity on a deterministic graph. Let $G$ be a deterministic graph and $\A$ the adjacency matrix of $G$ with columns normalized. We have the following definition of the {\em $k$th SimRank similarity matrix} $\Sim^{(k)}$.
\begin{align*}
	\small
	\Sim^{(0)} &= \I, &\\
	\Sim^{(n)} &= c\A^T\Sim^{(n - 1)}\A + (1 - c)\I & \text{for } n > 0,
\end{align*}
where $\I$ is the identity matrix, and $0 < c < 1$ is the delay factor. Since $G$ is a deterministic graph, we have $\A^k = \W^{(k)}$ for $k > 0$. Moreover, $(\A^T)^k = (\A^k)^T$. Let $\W^{(0)} = \I$. Thus,
\begin{equation*}\label{Eqn:SimRank-n}
	\small
	\Sim^{(n)} = c^n (\W^{(n)})^T \W^{(n)} + (1 - c) \sum_{k = 0}^{n - 1} c^k (\W^{(k)})^T \W^{(k)}.
\end{equation*}
Note that the element at the $i$th row and the $j$th column of $(\W^{(k)})^T \W^{(k)}$ is the probability of two random walks starting from vertices $i$ and $j$, respectively, meeting at the same vertex after exactly $k$ transitions. The theorem below states that $\Sim^{(n)}$ converges to the {\em SimRank similarity matrix} $\Sim$ as $n \to +\infty$.
\begin{theorem}
$\lim_{n \to +\infty} \Sim^{(n)} = \Sim$.
\end{theorem}

Based on the possible worlds model of uncertain graphs, we define SimRank similarity on an uncertain graph as follows.

\begin{definition}\label{Def:uSimRank}
For a uncertain graph $\G$, a delay factor $0 < c < 1$, and $n \ge 1$, the {\em $n$th SimRank similarity} between two vertices $u$ and $v$ in $\G$, denoted by $s^{(n)}_{\G}(u, v)$, is defined by
\begin{equation} \label{Eqn:kthSimRank}
	\small
	\begin{split}
		s^{(n)}_{\G}(u, v) &= c^n \sum_{w \in V(\G)} \Pr\nolimits_{\G}(u \to_n w) \Pr\nolimits_{\G}(v \to_n w)\\
		&+ (1 - c) \sum_{k = 0}^{n - 1} c^k \Pr\nolimits_{\G}(u \to_k w) \Pr\nolimits_{\G}(v \to_k w).
	\end{split}
\end{equation}
The {\em SimRank similarity} between $u$ and $v$, denoted by $s_{\G}(u, v)$, is defined by $s_{\G}(u, v) = \lim_{n \to +\infty} s^{(n)}_{\G}(u, v)$.
\end{definition}

The following theorem gives an upper bound on the error between $s^{(n)}_{\G}(u, v)$ and $s_{\G}(u, v)$. The proof is omitted.

\begin{theorem}
For $n \ge 1$, $|s^{(n)}_{\G}(u, v) - s_{\G}(u, v)| \le c^{n + 1}$.
\end{theorem}

As $n$ increases, the error between $s^{(n)}_{\G}(u, v)$ and $s_{\G}(u, v)$ decreases exponentially. Therefore, similar to SimRank computation on a deterministic graph, we can use $s^{(n)}_{\G}(u, v)$ as a good approximation of $s_{\G}(u, v)$ when $n$ is sufficiently large.

The following theorem shows that the SimRank similarity defined on uncertain graphs is a \emph{ generalization} of the SimRank similarity defined on deterministic graphs.

\begin{theorem}
Let $G$ be a deterministic graph and $\G$ be an uncertain graph with $V(\G) = V(G)$, $E(\G) = E(G)$, and $P_{\G}(e) = 1$ for all $e \in E(\G)$. For all $u, v \in V(\G)$, the SimRank similarity between $u$ and $v$ on $\G$ equals the SimRank similarity between $u$ and $v$ on $G$.
\end{theorem}

\section{Computing SimRank Similarities}
\label{sec:twostage}

In this section we propose several algorithms for computing the SimRank similarity between two vertices in an uncertain graph, namely the baseline algorithm in Section \ref{Sec:Baseline}, the sampling algorithm in Section \ref{Sec:Sampling}, the two-phase algorithm in Section \ref{Sec:TwoStage} and the speeding-up algorithm in Section~\ref{Sec:Speedup}.

\subsection{Baseline Algorithm}
\label{Sec:Baseline}

We first describe the baseline algorithm. Given as input an uncertain graph $\G$, two vertices $u, v$, a real number $c \in (0, 1)$ and an integer $n > 0$, we first compute the transition probability matrices $\UW^{(1)},\UW^{(2)}, \dots, \UW^{(n)}$ by the \textsf{TransPr} algorithm. Then, we compute $s_{\G}^{(n)}(u, v)$ by Eq.~\eqref{Eqn:kthSimRank} and return $s_{\G}^{(n)}(u, v)$ as an approximation of $s_{\G}(u, v)$.

To compute the term $\sum_{w \in V(\G)} \Pr_{\G}(u \to_k w) \Pr_{\G}(v \to_k w)$ in Eq.~\eqref{Eqn:kthSimRank}, we need to read the two columns of $\UW^{(k)}$ corresponding to $u$ and $v$, respectively. Since $\UW^{(k)}$ is not sparse, we store $\UW^{(k)}$ in external memory. To facilitate data access, we store the elements of $\UW^{(k)}$ column-by-column in consecutive blocks on disk. Let $B$ be the size of a disk block. Reading a column requires $O(|V(\G)|/B)$ I/O's. Hence, the total number of I/O's of the baseline algorithm is $O(n|V(\G)|/B)$.

\subsection{Sampling Algorithm}
\label{Sec:Sampling}

The second algorithm for computing $s_{\G}^{(n)}(u,v)$ is based on random sampling. In this algorithm, we estimate each term $\sum_{w \in V(\G)} \Pr_{\G}(u \to_k w) \Pr_{\G}(v \to_k w)$ in Eq.~\eqref{Eqn:kthSimRank} via sampling. For $k > 0$, the sampling procedure is as follows.

Let $N$ be a sufficiently large integer. We randomly sample $N$ walks $W^u_1, W^u_2, \ldots, W^u_N$ starting from $u$ and $N$ walks $W^v_1, W^v_2, \ldots, W^v_N$ starting from $v$, all of which are of length $n$. Each walk $W$ should be sampled with the walk probability of $W$. A simple method to do this is to first sample a possible world $G$ of $\G$ with probability $\Pr(\G \Rightarrow G)$ and then randomly select a walk $W$ on $G$. Indeed, this method is  inefficient. In our algorithm, we adopt a more efficient method. Without loss of generality, let us take $u$ as the starting vertex. For every vertex of $\G$, we record the status if the vertex has been visited by $W$. Initially, $W$ only contains $u$. Then, we extend $W$ by iteratively appending vertices at the end of $W$ until the length of $W$ reaches $k$. In each iteration, we check whether the last vertex $w$ in $W$ has been visited. If $w$ has not been visited, we first instantiate every arc $e$ leaving $w$ with probability $P_{\G}(e)$ and designate $w$ as visited. If $w$ has been visited, we select uniformly at random an out-neighbor $z$ of $w$ via an instantiated arc leaving $w$, and append $z$ at the end of $W$.

After sampling $N$ walks $W^u_1, W^u_2, \ldots, W^u_N$ starting from $u$ and $N$ walks $W^v_1, W^v_2, \ldots, W^v_N$ starting from $v$, we estimate $\sum_{w \in V(\G)} \Pr_{\G}(u \to_k w) \Pr_{\G}(v \to_k w)$ for all $k > 0$. For ease of presentation, let us denote $\sum_{w \in V(\G)} \Pr_{\G}(u \to_k w) \Pr_{\G}(v \to_k w)$ by $m^{(k)}(u, v)$. For $1 \le i \le N$ and $1 \le j \le n$, let $I(i, j) = 1$ if the $j$th vertex in $W^u_i$ and the $j$th vertex in $W^v_i$ are the same; otherwise, $I(i, j) = 0$. Hence, $m^{(k)}(u, v)$ can be estimated by
\begin{equation}\label{Eqn:MeetingProb}
	\small
	\widehat{m}^{(k)}(u, v) = \frac{1}{N} \sum_{i = 1}^N \sum_{j = 1}^n I(i, j).
\end{equation}

By Chernoff's bound, we have the following lemma on the error of $\widehat{m}^{(k)}(u, v)$.

\begin{lemma}
For $\epsilon > 0$ and $0 < \delta < 1$, if $N \ge \frac{3}{\epsilon^2} \ln \frac{2}{\delta}$, then
\begin{equation*} \label{Eqn:SimAbError}
	\small
	\Pr(|m^{(k)}(u, v) - \widehat{m}^{(k)}(u, v)| \leq \epsilon) \geq 1 - \delta.
\end{equation*}
\end{lemma}
Consequently, by Eq.~\eqref{Eqn:kthSimRank}, we can estimate $s_{\G}^{(n)}(u,v)$ by
\begin{equation}\label{Eqn:SimRankSampling}
	\small
	\widehat{s}^{(n)}_{\G}(u, v) = c^n \widehat{m}^{(n)}(u, v)
	+ (1 - c) \sum_{k = 0}^{n - 1} c^k \widehat{m}^{(k)}(u, v).
\end{equation}

\begin{figure}[!t]
    \scriptsize
    \fbox{
    \parbox{\columnwidth}{
    \textbf{Algorithm} \textsf{Sampling($\G$, $u$, $v$, $n$, $N$)}
    \begin{algorithmic}[1]
        \FOR{$i \gets 1$ \TO $N$}
            \STATE $W_{i}^{u} \gets u$
            \FOR{$j \gets 1$ \TO $n$}
                \STATE // let $w$ be the last vertex of $W_{i}^{u}$
                \IF{$w$ is not visited by $W_{i}^{u}$}
                    \STATE sample all arcs leaving $w$ according to its existing probability
                \ENDIF
                \STATE mark $w$ visited by $W_{i}^{u}$
                \STATE choose an instantiated neighbor $x$ of $w$ at random
                \STATE append $x$ onto $W_{i}^{u}$
            \ENDFOR
        \ENDFOR
        \FOR{$i \gets 1$ \TO $N$}
            \STATE $W_{i}^{v} \gets v$
            \FOR{$j \gets 1$ \TO $n$}
                \STATE // let $w$ be the last vertex of $W_{i}^{v}$
                \IF{$w$ is not visited by $W_{i}^{v}$}
                    \STATE sample all arcs leaving $w$ according to its existing probability
                \ENDIF
                \STATE mark $w$ visited by $W_{i}^{v}$
                \STATE choose an instantiated out-neighbor $z$ of $w$ at random
                \STATE append $z$ at the end of $W_{i}^{v}$
            \ENDFOR
        \ENDFOR
        \FOR{$k \gets 1$ \TO $n$}
            \STATE compute $\widehat{m}^{(k)}(u, v)$ according to Eq.~\eqref{Eqn:MeetingProb}
        \ENDFOR
        \STATE compute $\widehat{s}^{(n)}_{\G}(u, v)$  Eq.~\eqref{Eqn:SimRankSampling}
        \RETURN $\widehat{s}^{(n)}_{\G}(u, v)$
    \end{algorithmic}
    }}
    \caption{Algorithm \textsf{\small Sampling}.}
    \label{Fig:Sampling}
\end{figure}

The \textsf{Sampling} algorithm in Fig.~\ref{Fig:Sampling} illustrates the details of our sampling algorithm. We have the following result on the approximation error of \textsf{Sampling}.

\begin{theorem} \label{Thm:RelBound}
For $\epsilon > 0$ and $0 < \delta < 1$, if $N \ge \frac{3}{\epsilon^2} \ln \frac{2}{\delta}$,
\begin{equation*}
	\small
	\Pr\left(|s^{(n)}_{\G}(u, v) - \widehat{s}^{(n)}_{\G}(u, v)| \leq (c - c^n)\epsilon\right)\geq  1 - \delta.
\end{equation*}
\end{theorem}

The time complexity of {\sf Sampling} is $O(Nnd)$, where $d$ is the average degree of the vertices in $\G$. The \textsf{Sampling} algorithm is more efficient than the baseline algorithm because it performs less I/O's and uses less memory.

\subsection{The Two-stage Algorithm}
\label{Sec:TwoStage}

To take the advantages of both the baseline algorithm and the sampling algorithm, we propose the two-stage algorithm. The algorithm is based on the following observations.

(1) When $k$ is small, the number of nonzero elements in the transition probability matrix $\UW^{(k)}$ is far less than $|V(\G)|^2$. Especially, there are only $|E(\G)|$ nonzero elements in $\UW^{(1)}$. It decreases the number of I/O's to read the columns of $\UW^{(k)}$. Thus, the method used in the baseline algorithm is efficient in computing the exact value of $m^{(k)}(u, v)$.

(2) When $k$ is large, the error of the estimated value $\widehat{m}^{(k)}(u, v)$ computed by the sampling method is less than $c^k \epsilon$ with probability at least $1 - \delta$.

Inspired by the two observations, we compute $s^{(n)}_{\G}(u, v)$ in two stages. Let $1 < l < n$.

\noindent{\bf Stage 1.} For $1 \le k \le l$, compute $m^{(k)}(u, v)$ using the exact method given in the baseline algorithm.

\noindent{\bf Stage 2.} For $l < k \le n$, estimate $m^{(k)}(u, v)$ using the method given in the \textsf{Sampling} algorithm. Let $\widehat{m}^{(k)}(u, v)$ be the estimated value of $m^{(k)}(u, v)$.

After the above two stages, we estimate $s^{(n)}_{\G}(u, v)$ by
\begin{equation}\label{Eqn:SimRankTwoPhase}
	\small
    \begin{split}
    	\widehat{s}^{(n)}_{\G}(u, v) &= c^n \widehat{m}^{(n)}(u, v)
    	+ (1 - c) \sum_{k = l + 1}^{n - 1} c^k \widehat{m}^{(k)}(u, v) \\
    	&\quad+ (1 - c) \sum_{k = 0}^{l} c^k m^{(k)}(u, v).
    \end{split}
\end{equation}

By Theorem~\ref{Thm:RelBound}, we immediately have the corollary below.
\begin{corollary}
\label{Cor:Error}
For $\epsilon > 0$ and $0 < \delta < 1$, if $N \ge \frac{3}{\epsilon^2} \ln \frac{2}{\delta}$,
\begin{equation*}
	\small
	\Pr\left(|s^{(n)}_{\G}(u, v) - \widehat{s}^{(n)}_{\G}(u, v)| \leq (c^{l + 1} - c^n)\epsilon\right)\geq  1 - \delta.
\end{equation*}
\end{corollary}

By tuning $l$, we can make a tradeoff between the time efficiency and the approximation error. By Corollary~\ref{Cor:Error}, the relative error of the output of the two-stage algorithm is bounded by $(c^{l+1} - c^n)\epsilon/s^{(n)}_{\G}(u, v)$. When $l$ is larger, the relative error of $s^{(n)}_{\G}(u, v)$ decreases exponentially. Meanwhile, the execution time of the two-stage algorithm  increases because we need to read more components of the matrices on disk. If $l$ is carefully chosen such that $\UW^{(1)}, \UW^{(2)}, \dots, \UW^{(l)}$ fit into main memory, the two-stage algorithm is as efficient as the \textsf{Sampling} algorithm. For example, let $l = 1$, $\UW^{(1)}$ only consumes $O(|E|)$ space. In this setting, for an SimRank similarity value which is about $c/10$, the relative approximation error is near $c\epsilon$, which is an order of magnitude better than the \textsf{Sampling} algorithm.

\subsection{Speeding-up Technique}
\label{Sec:Speedup}

\begin{figure}[t]
\scriptsize
\fbox{
\parbox{\columnwidth}{
\textbf{Algorithm} \textsf{Speedup($\G$, $u$, $v$, $n$, $N$)}
\begin{algorithmic}[1]
    \STATE $U^{(0)} \gets \{u\}$, $V^{(0)} \gets \{v\}$
    \STATE insert $(0,\mathbf{1})$ into $\M_u$ and $\M'_v$
    \FOR{$k \gets 0$ \TO $n-1$}
        \STATE $U^{(k+1)} \gets \emptyset$
        \FOR{each vertex $w$ in $U^{(k)}$}
            \FOR{each vertex $x$ in $O_{\G}(w)$}
                \STATE $\M_x[k+1] = \M_x[k+1] \vee (\M_w[k] \wedge \F_(w,x))$
                \IF{$\M_x[k+1] \neq \mathbf{0}$}
                    \STATE insert $x$ into $U^{(k+1)}$
                \ENDIF
            \ENDFOR
        \ENDFOR
    \ENDFOR
    \FOR{$k \gets 0$ \TO $n-1$}
        \STATE $V^{(k+1)} \gets \emptyset$
        \FOR{each vertex $w$ in $V^{(k)}$}
            \FOR{each vertex $x$ in $O_{\G}(w)$}
                \STATE $\M'_x[k+1] = \M'_x[k+1] \vee (\M'_w[k] \wedge \F_(w,x))$
                \IF{$\M'_x[k+1] \neq \mathbf{0}$}
                    \STATE insert $x$ into $V^{(k+1)}$
                \ENDIF
            \ENDFOR
        \ENDFOR
    \ENDFOR
    \FOR{$k \gets 1$ \TO $n$}
        \STATE compute $\widehat{m}^{(k)}(u, v)$ according to Eq.~\eqref{Eqn:SpeedupCount}
    \ENDFOR
    \STATE compute $\widehat{s}^{(n)}_{\G}(u, v)$ according to Eq.~\eqref{Eqn:SimRankTwoPhase}
    \RETURN $\widehat{s}^{(n)}_{\G}(u, v)$
\end{algorithmic}
}}
\caption{Algorithm \textsf{\small Speedup}.}
\label{Fig:Speedup}
\end{figure}

We now propose the technique for speeding up the sampling process in the sampling algorithm and the two-phase algorithm. Remember that, given two vertices $u$ and $v$, we sample independently at random $N$ walks $W^u_1, W^u_2, \ldots, W^u_N$ starting from $u$ and $N$ walks $W^v_1, W^v_2, \ldots, W^v_N$ starting from $v$, all of which are of length $n$. As analyzed in Section \ref{Sec:Sampling}, the expected time to sample each of these walks is $O(nd)$, where $d$ is the average degree of the vertices of $\G$. To reduce the total time of sampling all these walks, we utilize an observation that $W^u_1, W^u_2, \ldots, W^u_N$ (or $W^v_1, W^v_2, \ldots, W^v_N$) usually have significant overlaps, which can be used to reduce redundant extensions of walks. For example, Fig.~\ref{Tab:ExampleSamples} shows ten walks starting from $v_1$ and $v_2$, respectively, sampled from the uncertain graph shown in Fig.~\ref{Fig:UGraph:UG}. Notice that $v_3$ occurs five times in the first step of the five walks starting from $v_1$, and $v_4$ occurs twice in the first step of the fives walks starting from $v_2$. If we could compress the traverses of these extensions of walks, a lot of redundant operations could be reduced, which will further speed up the algorithm.

\begin{figure}
	\centering\scriptsize
    \begin{tabular}{c|c||c|c} \hline
         & Walks starting from $v_1$ & & Walks starting from $v_2$ \\ \hline
        $W_1$ & $v_1,v_3,v_1,v_3,v_1,v_3$ & $W_6$ & $v_2,v_1,v_3,v_4,v_5,v_3$ \\
        $W_2$ & $v_1,v_3,v_4,v_5,v_3,v_4$ & $W_7$ & $v_2,v_3,v_1,v_3,v_4,v_5$ \\
        $W_3$ & $v_1,v_3,v_4,v_5,v_3,v_1$ & $W_8$ & $v_2,v_4,v_5,v_3,v_1,v_3$ \\
        $W_4$ & $v_1,v_3,v_1,v_3,v_4,v_5$ & $W_9$ & $v_2,v_3,v_4,v_5,v_3,v_1$ \\
        $W_5$ & $v_1,v_3,v_4,v_5,v_3,v_4$ & $W_{10}$ & $v_2,v_4,v_5,v_3,v_1,v_3$ \\ \hline
    \end{tabular}
    \caption{An example of ten walks of length $5$ sampled at random.}
    \label{Tab:ExampleSamples}
\end{figure}

In our speeding-up method, we associate every vertex $w$ of $\G$ with a hash table, denoted by $\M_w$. We call $\M_w$ the {\em counting table} of $w$. The key of each entry of the hash table is an integer. For each key $k$, the value corresponding to $k$ is a $N$-dimensional bit vector, denoted by $\M_w[k]$. For $1 \le i \le N$, the $i$th bit of $\M_w[k]$ is $1$ if and only if $w$ is the $k$th vertex in the walk $W^u_i$. If the hash table $\M_w$ does not contain key $k$, then we conceptually have $\M_w[k] = \mathbf{0}$. Similarly, we associate every vertex $w$ of $\G$ with a hash table $\M'_w$. For each key $k$, the $i$th bit of $\M'_w[k]$ is $1$ if and only if $w$ is the $k$th vertex in the walk $W^v_i$, where $1 \le i \le N$. Thus, the information of the walks $W^u_1, W^u_2, \ldots, W^u_N$ (or $W^v_1, W^v_2, \ldots, W^v_N$) is {\em losslessly} encoded in all nonempty hash tables $\M_w$ (or $\M'_w$).

Moreover, we associate every arc $e = (w, x)$ of $\G$ with a $N$-dimensional bit vector $\F_{e}$, called the {\em filter vector} of $e$. We construct the filter vectors of all the arcs of $\G$ offline. Initially, we set the bit vectors of all arcs to $\mathbf{0}$, where $\mathbf{0}$ is a vector with all bits set to $0$. For all vertices $w$ of $\G$ and all $1 \le i \le N$, we first instantiate every arc $e = (w, x)$ with probability $P_{\G}(e)$, where $x \in O_{\G}(w)$. Then, we select one instantiated arc $e$ leaving $w$ uniformly at random and set the $i$th bit of the filter vector $\F_e$ to $1$.

We now describe our method for speeding up the sampling process. Suppose that the filter vectors of all edges have been constructed offline. In our new method, to obtain the sampled walks $W^u_1, W^u_2, \ldots, W^u_N$, we need not to perform the sampling process $N$ times. Instead, we start from vertex $u$ and perform $N$ sampling processes simultaneously by leveraging the common substructures among samples. The \textsf{Speedup} algorithm in Figure \ref{Fig:Speedup} describes the process of the speedup algorithm. The details of the method are given as follows.

\noindent{\bf Step 1.} Let $U^{(k)}$ be the set of vertices that are probable to be visited at the $k$th step of a walk starting from $u$. Initially, we set $U^{(0)} = \{u\}$ and $U^{(k)} = \emptyset$ for all $1 \le k \le n$. Since $u$ is the $0$th vertex in all sampled walks, we insert an entry $(0, \mathbf{1})$ to the counting table $\M_u$, where $\mathbf{1}$ is the bit vector with all bits set to $1$. For $k$ from $0$ to $n - 1$, we perform the $k$th iteration as follows. For each vertex $w \in U^{(k)}$, we retrieve the set, $O_{\G}(w)$, of all out-neighbors of $w$. Each out-neighbor $x$ of $w$ is probable to be visited at the $(k + 1)$-th step of a walk starting from $u$. Recall that the bit vector $\M_w[k]$ records that in which sampled walks, $w$ is the $k$th vertex. Suppose the $i$th bit of $\M_w[k]$ is $1$, that is, $w$ is the $k$th vertex in the $i$th sampled walk $W^u_i$. For every out-neighbor $x$ of $w$, if the $i$th bit of the filter vector $\F_{(w, x)}$ is $1$, that is, the walk goes from $w$ to $x$ in $W_i^u$, then $x$ is certainly the $(k + 1)$-th vertex in $W_i^u$, that is, the $i$th bit of $\M_x[k + 1]$ should be set to $1$. Thus, we update $\M_x[k + 1]$ to $\M_x[k + 1] \vee (\M_w[k] \wedge \F_{(w, x)})$, where $\vee$ and $\wedge$ denote the bit-wise OR and the bit-wise AND operations, respectively. If $\M_x[k + 1] \ne \mathbf{0}$, we insert $x$ to $U^{(k + 1)}$.

\noindent{\bf Step 2.} Let $V^{(k)}$ be the set of vertices that are probable to be visited at the $k$th step of a walk starting from $u$. Similar to Step 1, we start from vertex $v$ and perform another $n$ iterations to compute $V^{(k)}$ for $1 \leq k \leq n$ and update the counting tables $\M'_w$ associated with the vertices.

\noindent{\bf Step 3.} We compute $m^{(k)}(u, v)$ based on the counting tables $\M_w$ and $\M'_w$. In particular, for $0 \le k \le n$, the value of $m^{(k)}(u, v)$ can be estimated by
\begin{equation} \label{Eqn:SpeedupCount}
	\small
	\widehat{m}^{(k)}(u, v) = \frac{1}{N} \sum_{w \in U^{(k)} \cap V^{(k)}} \|\M_w[k] \wedge \M'_w[k]\|_1,
\end{equation}
where $\|\mathbf{x}\|$ denotes the $1$-norm of a vector $\mathbf{x}$, that is, the number of $1$'s in $\mathbf{x}$. It is easy to verify that the value $\widehat{m}^{(k)}(u, v)$ computed by the above equation is the same as the one computed by Eq.~\eqref{Eqn:MeetingProb}. Consequently, we can estimate $s_{\G}^{(n)}(u, v)$ either by Eq.~\eqref{Eqn:SimRankSampling} or by Eq.~\eqref{Eqn:SimRankTwoPhase}.

\section{Experiments}
\label{Sec:Exp}

We conducted extensive experiments to evaluate the effectiveness of our SimRank similarity measure and the performance of the proposed algorithms. We present the experimental results in this section.

\subsection{Experimental Setting}

We implemented the proposed algorithms in C++, including the baseline algorithm ({\sf Baseline}), the sampling algorithm ({\sf Sampling}), the two-stage algorithm ({\sf SR-TS}) and the two-stage algorithm with the speed-up technique ({\sf SR-SP}). All experiments were run on a Linux machine with 2.5GHz Intel Core i5 CPU and 16GB of RAM.

Table~\ref{Tab:datasets} summarizes the datasets used in our experiments. {\em PPI1}, {\em PPI2} and {\em PPI3} are three protein-protein interaction networks extracted from \cite{kollios2013clustering} and the STRING\footnote{\st \tt http://string-db.org} database. {\it Condmat} and {\it Net} are two widely used co-authorship networks provided by \cite{newman2006finding}. {\it DBLP} is a co-authorship obtained from AMiner\footnote{\st  \tt http://aminer.org/citation}. For {\it Condmat}, {\it Net} and {\it DBLP}, we set the uncertainty of each edge using the method in \cite{zou2013structural}.

If not otherwise stated, on each uncertain graph, we ran the algorithms $1000$ times. For each time, we selected a pair of vertices  uniformly at random from the input uncertain graph and computed the SimRank similarity between the vertices. We evaluated the average execution time and relative error of the $1000$ runs. Therefore, the execution time and the relative error reported in this section are actually the average execution time and the average relative error, respectively. By default, we set $n=5$, $c = 0.6$  and $N =1000$. As reported in the following, the SimRank similarity generally converges within $5$ iterations.

\begin{table}[t]
    \centering\scriptsize
    \caption{Summary of datasets used in experiments.}
    \begin{tabular}{lrrr}
    	\hline
    	{Dataset}	& Number of vertices & Number of edges\\
    	\hline
    	{\em PPI1}	& 2708 & 7123 \\
    	{\em PPI2}	& 2369 & 249080 \\
        {\em PPI3}	& 19247 & 17096006 \\
        {\em Condmat} & 31163 & 240058 \\
        {\em Net}     & 1588 & 5484 \\
        {\em DBLP}    &1560640 &8517894 \\
        \hline
    \end{tabular}
    \label{Tab:datasets}
\end{table}

\subsection{Experimental Results}

\noindent\underline{\bf Differences between Similarity Measures.} First, we test the effectiveness of our SimRank similarity measure by comparing the similarities computed using our SimRank similarity measure and those computed using other similarity measures. We choose $1000$ pairs of vertices on {\it Net} and {\it PPI1} uniformly at random. For each pair of vertices, we compute their SimRank similarity by the {\sf Baseline} algorithm ({\sf SimRank-I}) and compare it with the similarities computed by other methods, including the SimRank similarity computed on the deterministic graph obtained by removing uncertainty from the uncertain graph ({\sf SimRank-II}), the SimRank similarity computed by Du et al.'s algorithm \cite{du2015probabilistic} ({\sf SimRank-III}), the Jaccard similarity computed by the algorithm in \cite{zou2013structural} ({\sf Jaccard-I}) and the Jaccard similarity computed on the deterministic graph obtained by removing uncertainty from the uncertain graph ({\sf Jaccard-II}).

Fig.~\ref{Fig:Difference} reports the differences between {\sf SimRank-I} and the other similarities. Fig.~\ref{Fig:AbsoluteNet} and Fig.~\ref{Fig:AbsolutePPI1} show the {\sf SimRank-I} similarities of $1000$ randomly selected vertex pairs on {\it Net} and {\it PPI1}, respectively. The vertex pairs are sorted in decreasing order of their {\sf SimRank-I} similarities. Fig.~\ref{Fig:ExactDGNet}--\ref{Fig:ExactJacDGNet} compare {\sf SimRank-I}  with {\sf SimRank-II}, {\sf SimRank-III}, {\sf Jaccard-I} and {\sf Jaccard-II} computed on {\it Net}, respectively. Fig.~\ref{Fig:ExactDGPPI1}--\ref{Fig:ExactJacDGPPI1} compare {\sf SimRank-I} with {\sf SimRank-II}, {\sf SimRank-III}, {\sf Jaccard-I} and {\sf Jaccard-II} computed on {\it PPI1}, respectively. Note that all similarities are normalized to within $[0, 1]$. The differences are summarized in Table~\ref{Tab:Biases}. We observe that (1) when {\sf SimRank-I} varies slightly, the other similarities may vary significantly; (2) when {\sf SimRank-I} decreases, the other similarities may increase. The differences between {\sf SimRank-I} and {\sf Jaccard-I} and {\sf Jaccard-II} are most significant because the Jaccard similarity cannot measure the similarity between vertices without common neighbors. {\sf SimRank-I} is  different from {\sf SimRank-II} since {\sf SimRank-II} does not consider uncertainty in graphs. {\sf SimRank-I} also differs from {\sf SimRank-III} since {\sf SimRank-III} is based on an unreasonable assumption as we mentioned in Section \ref{sec:RWCom}.

\begin{figure*}[t]
    \centering
    \scriptsize
    \subfigure[]{\includegraphics[width=0.19\linewidth]{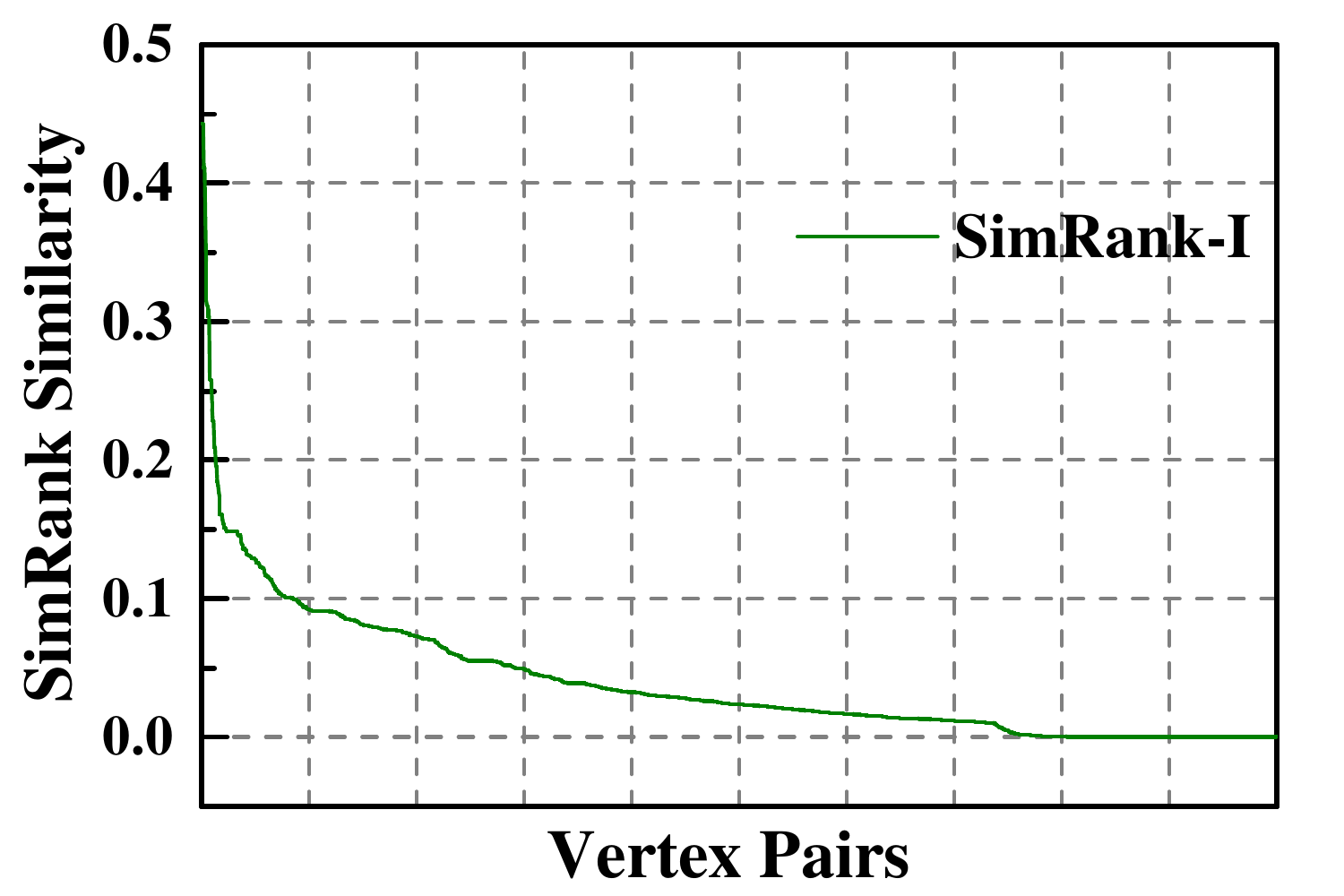} \label{Fig:AbsoluteNet}}
    \subfigure[]{\includegraphics[width=0.19\linewidth]{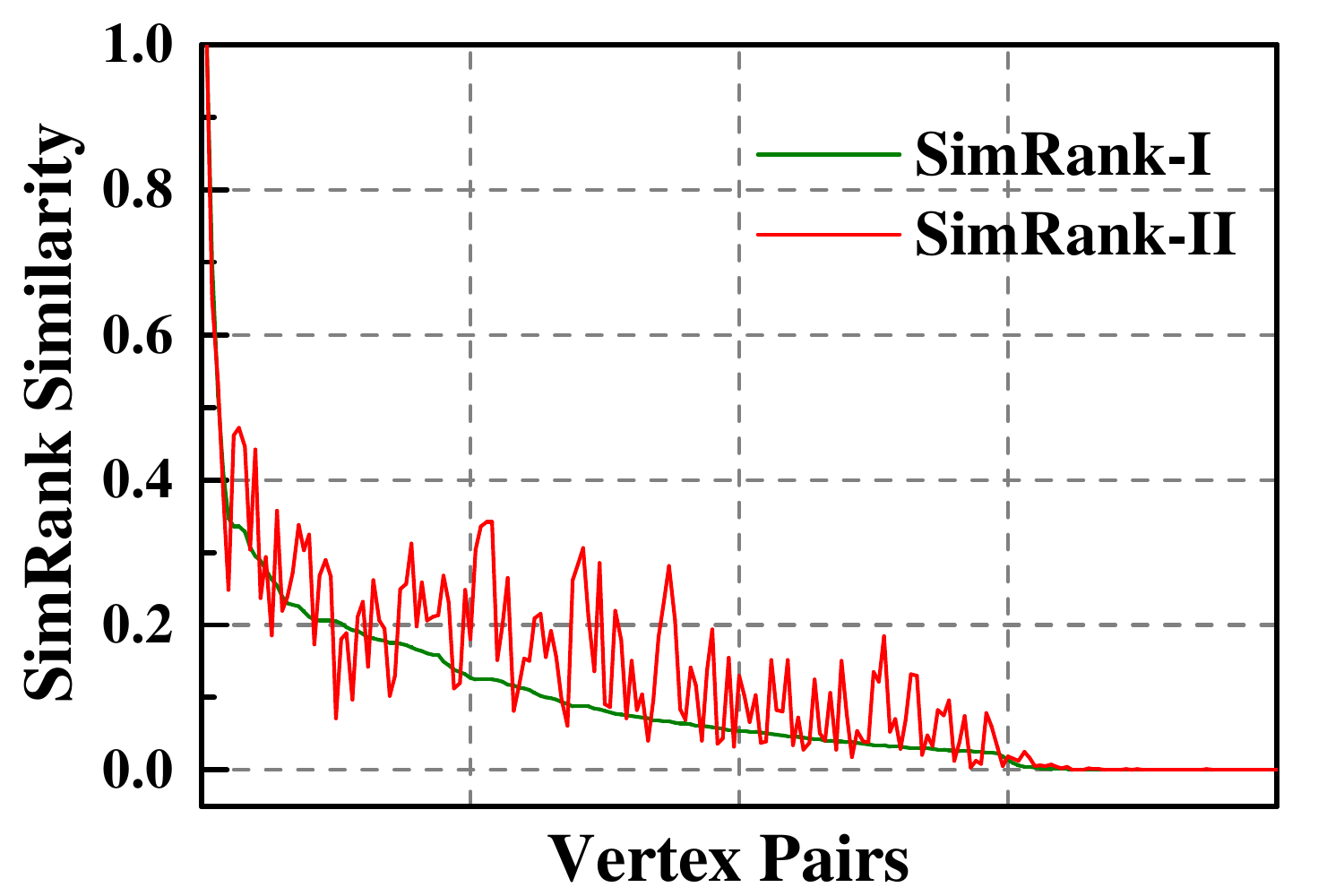} \label{Fig:ExactDGNet}}
    \subfigure[]{\includegraphics[width=0.19\linewidth]{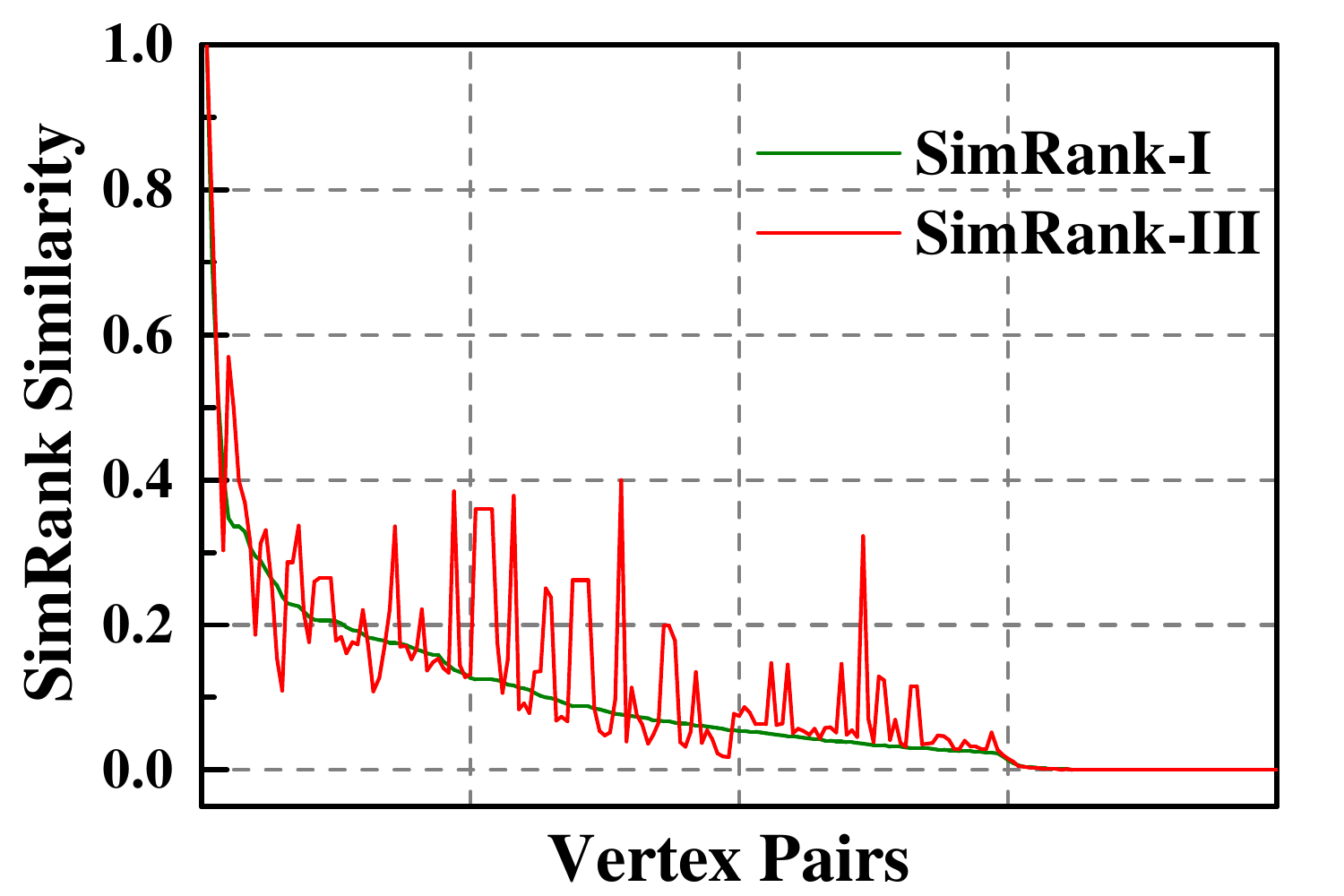}}
    \subfigure[]{\includegraphics[width=0.19\linewidth]{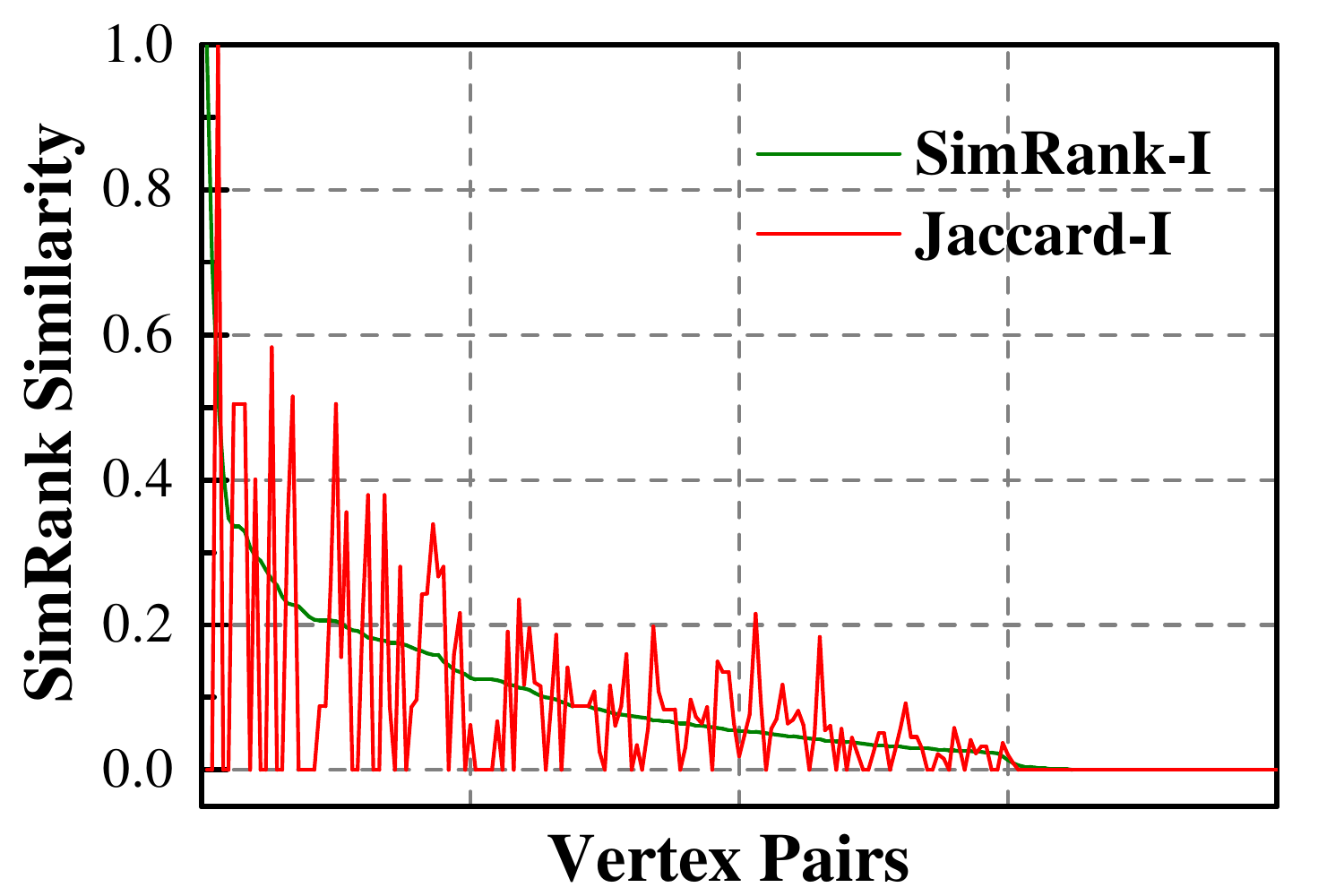}}
    \subfigure[]{\includegraphics[width=0.19\linewidth]{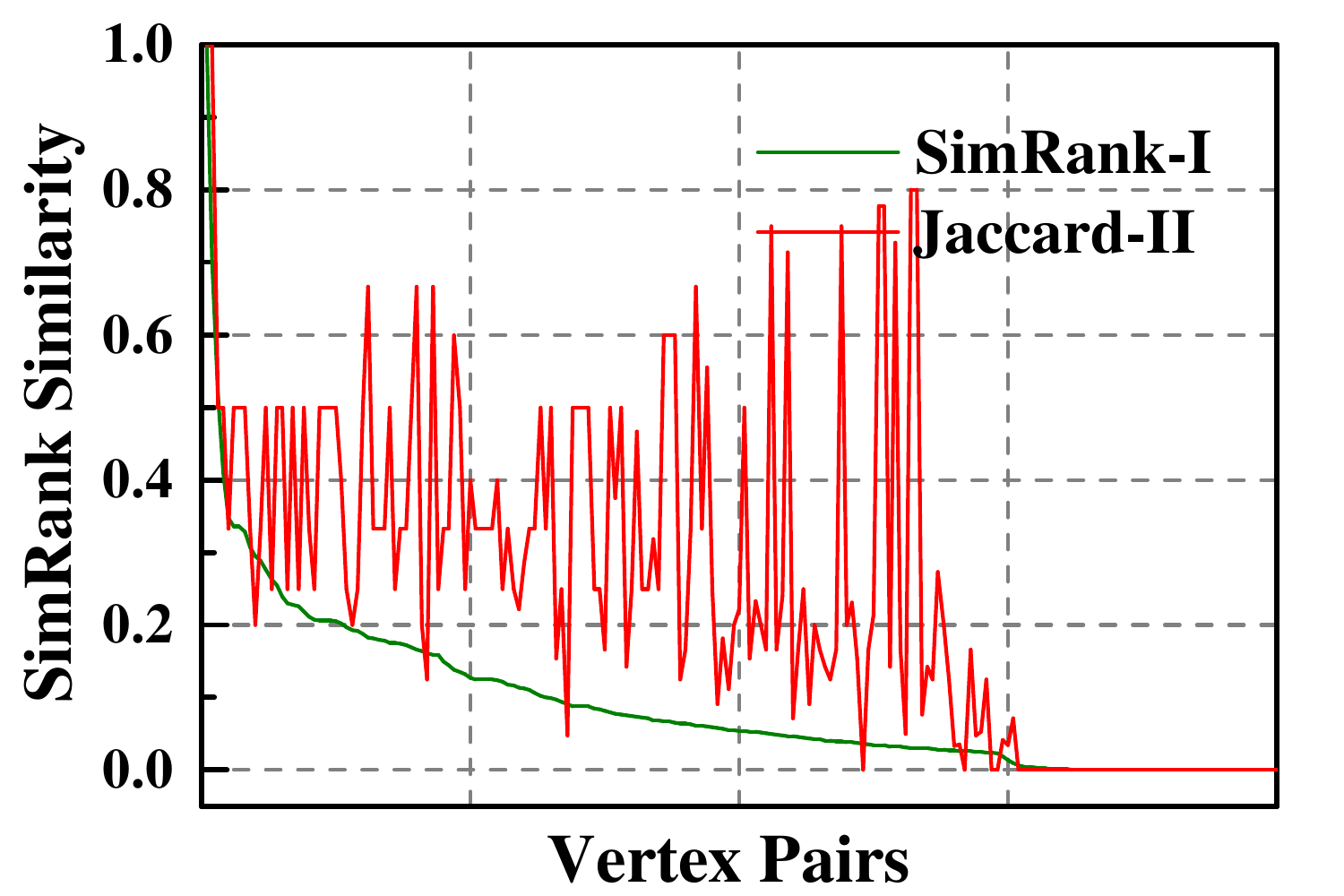} \label{Fig:ExactJacDGNet}}
    \subfigure[]{\includegraphics[width=0.19\linewidth]{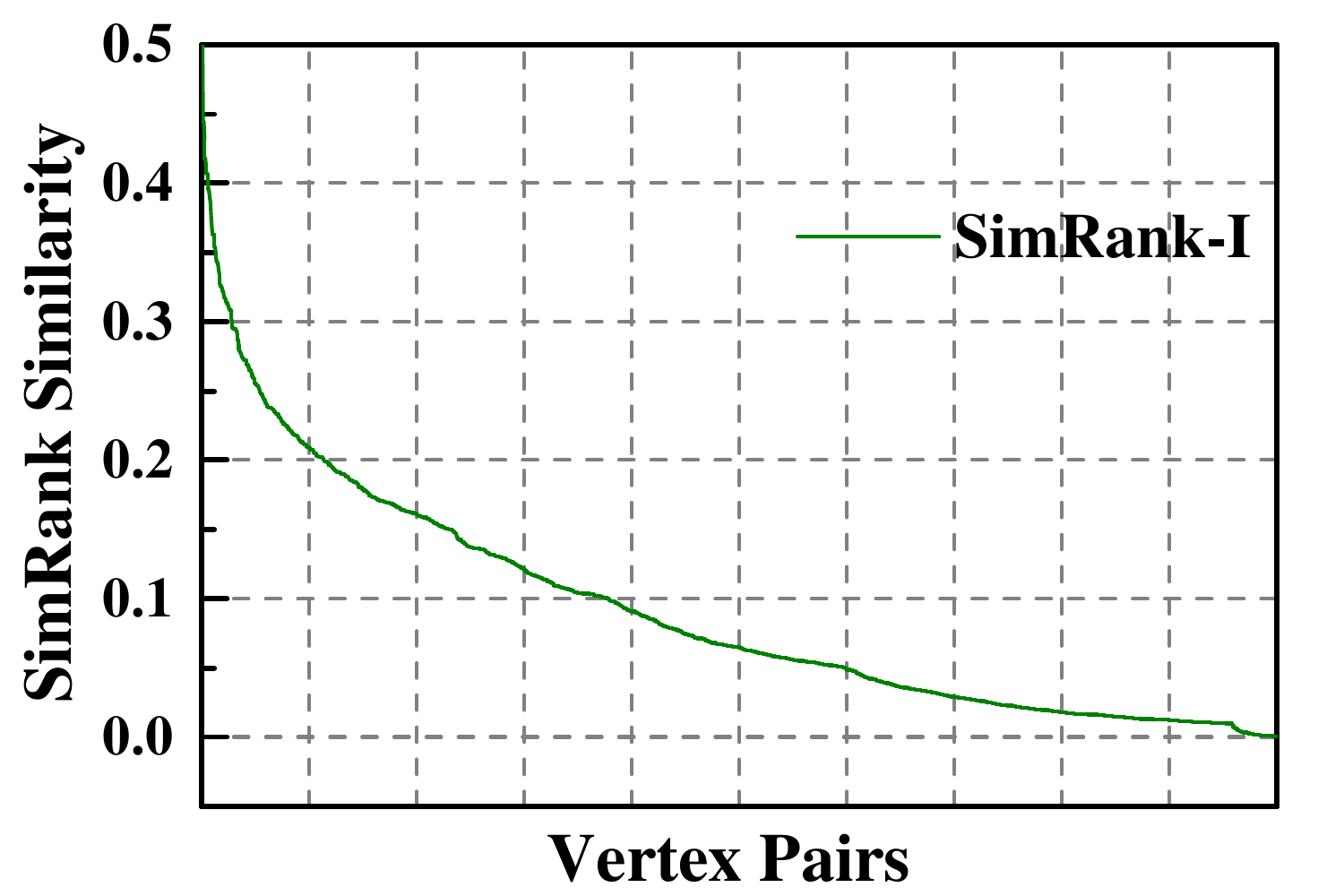} \label{Fig:AbsolutePPI1}}
    \subfigure[]{\includegraphics[width=0.19\linewidth]{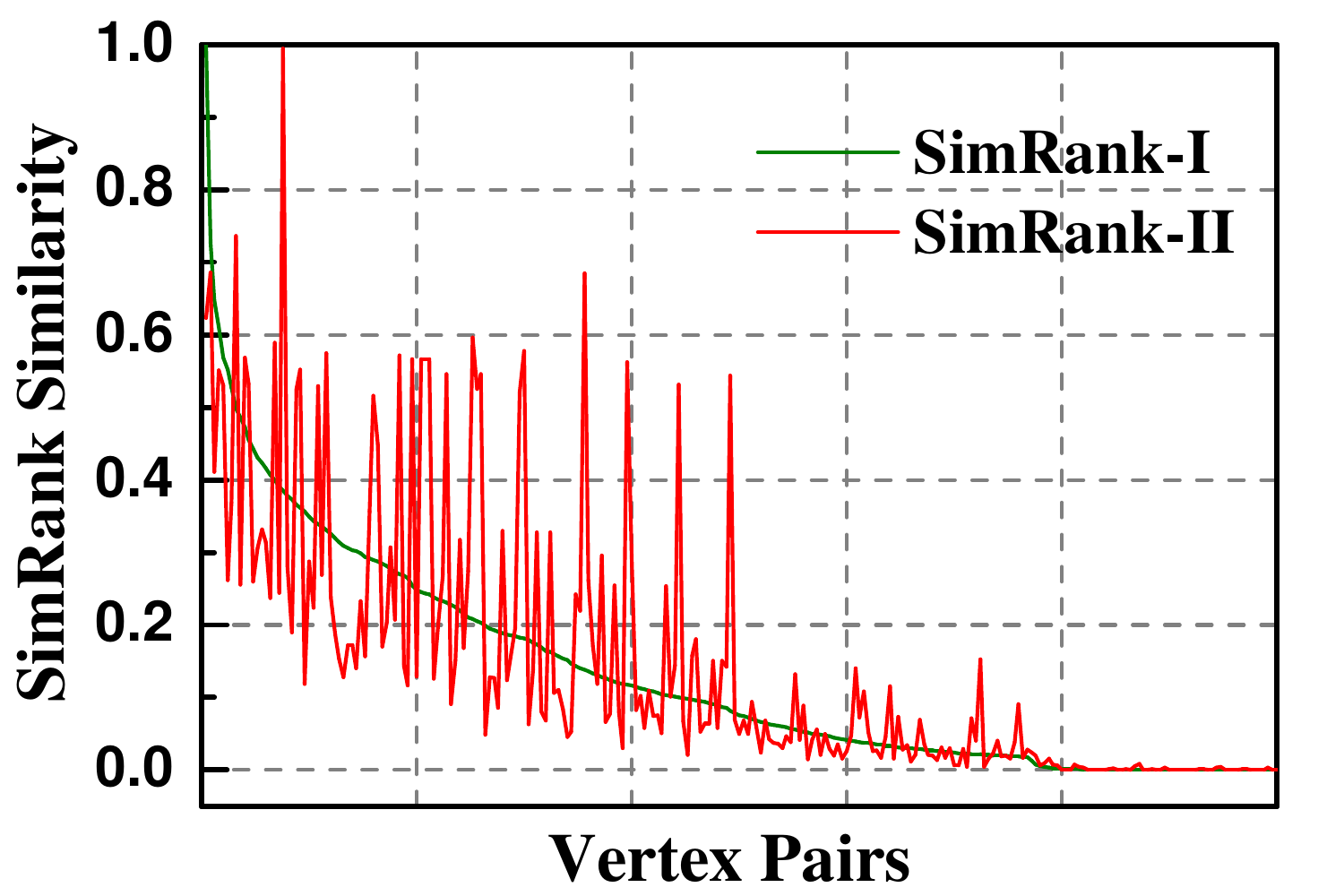} \label{Fig:ExactDGPPI1}}
    \subfigure[]{\includegraphics[width=0.19\linewidth]{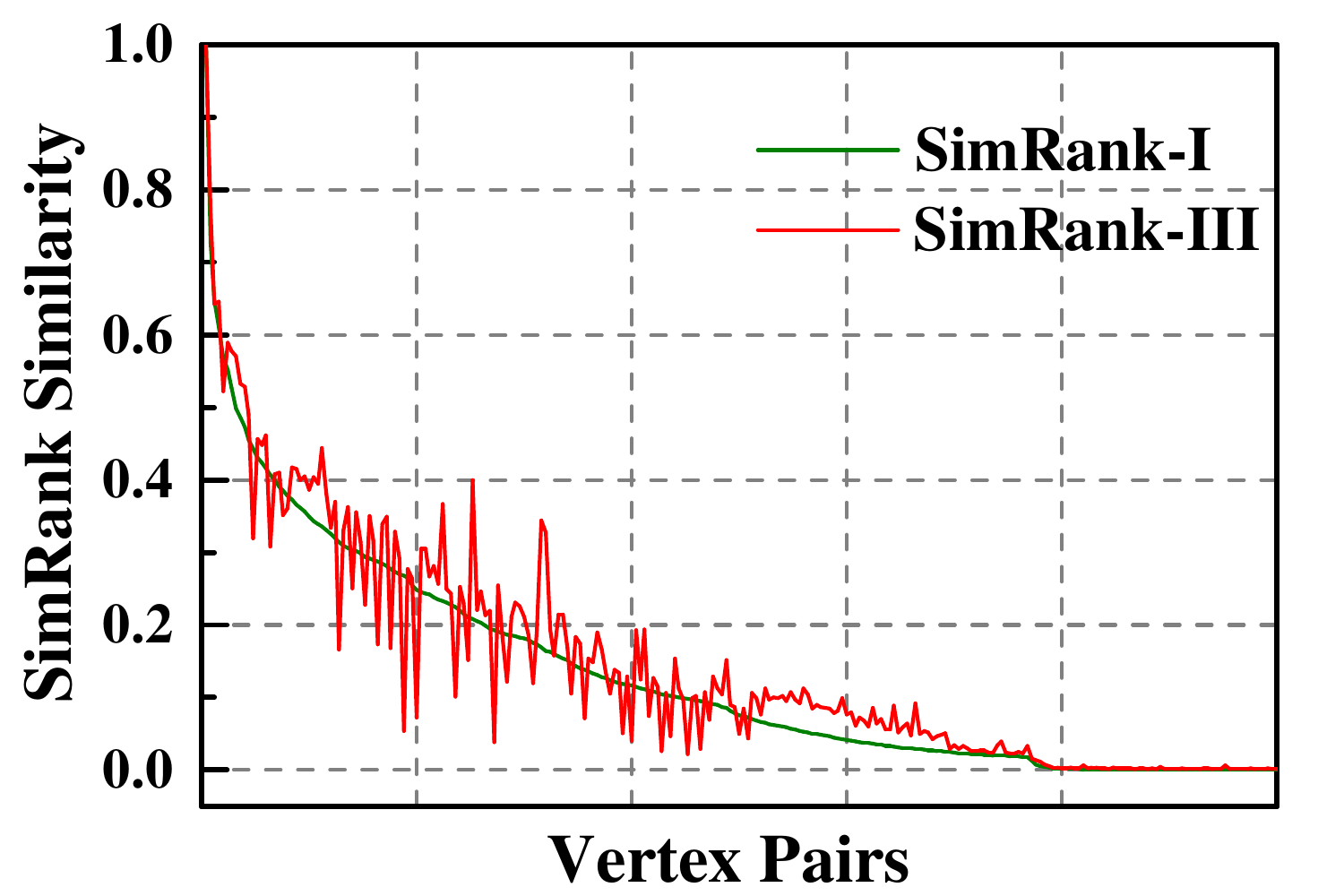}}
    \subfigure[]{\includegraphics[width=0.19\linewidth]{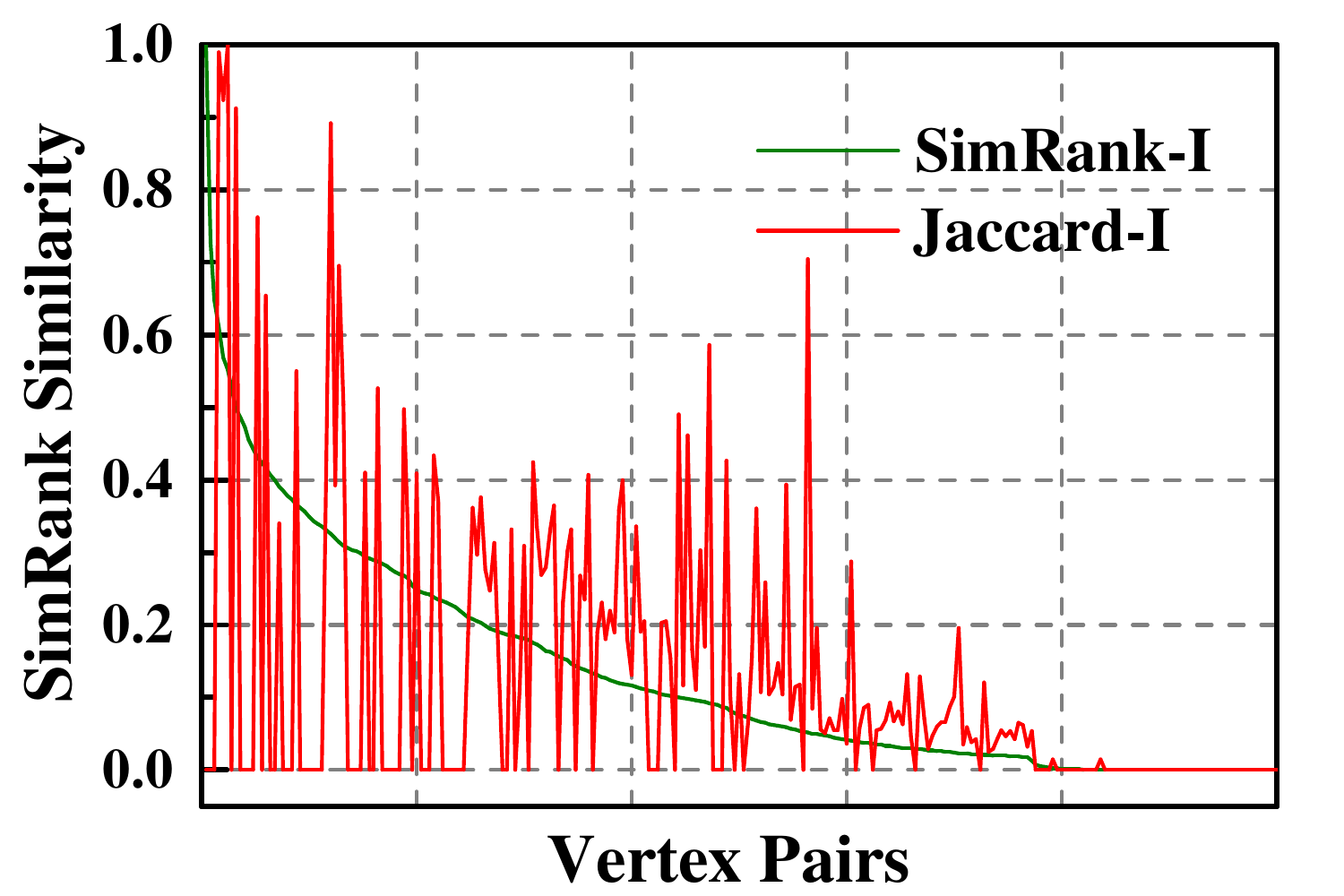}}
    \subfigure[]{\includegraphics[width=0.19\linewidth]{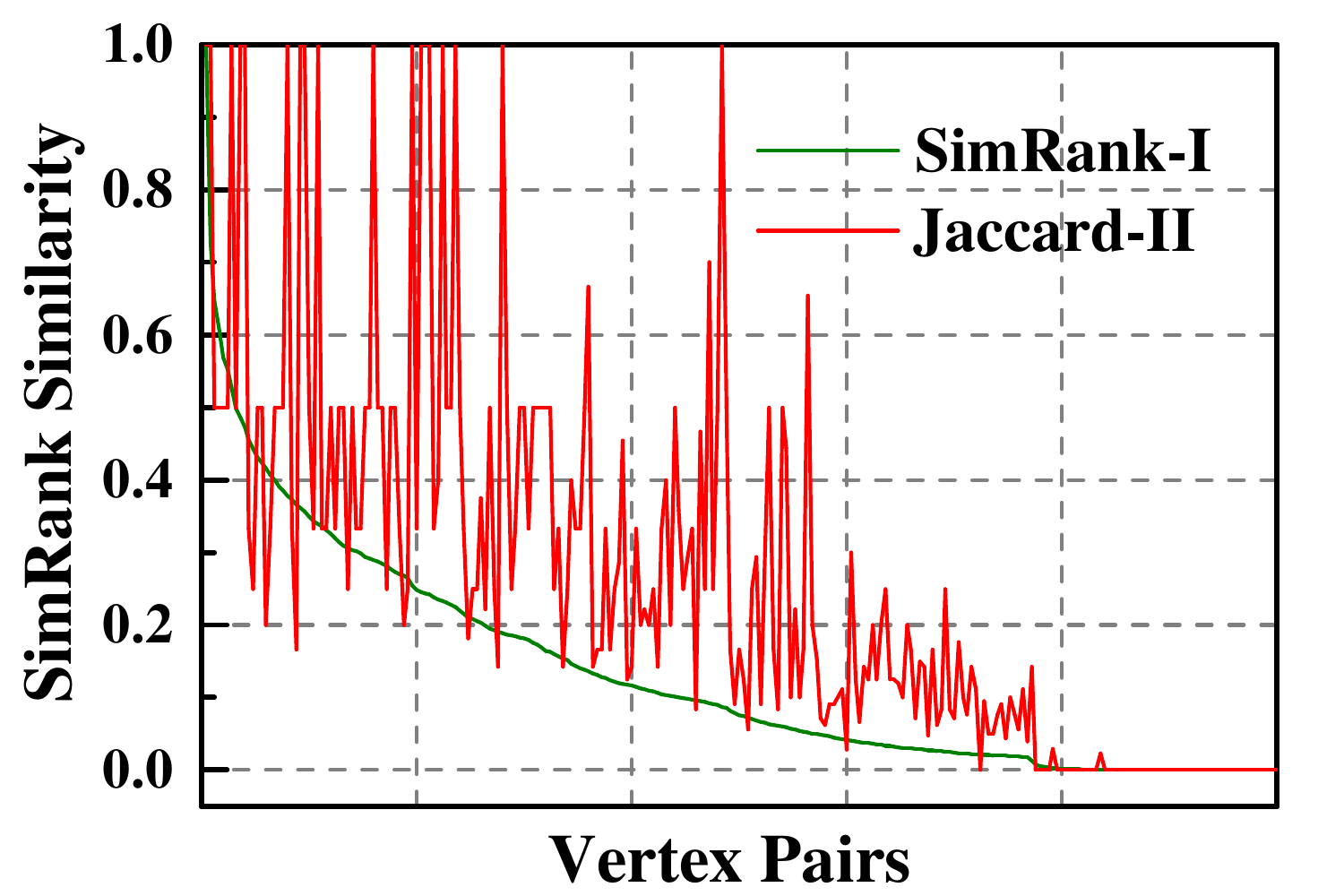} \label{Fig:ExactJacDGPPI1}}
    \caption{Differences between similarity measures. (a) {\sf SimRank-I} on {\it Net}. (b) {\sf SimRank-I} vs. {\sf SimRank-II} on {\it Net}. (c) {\sf SimRank-I} vs. {\sf SimRank-III} on {\it Net}. (d) {\sf SimRank-I} vs. {\sf Jaccard-I} on {\it Net}. (e) {\sf SimRank-I} vs. {\sf Jaccard-II} on {\it Net}. (f) {\sf SimRank-I} on {\it PPI1}. (g) {\sf SimRank-I} vs. {\sf SimRank-II} on {\it PPI1}. (h) {\sf SimRank-I} vs. {\sf SimRank-III} on {\it PPI1}. (i) {\sf SimRank-I} vs. {\sf Jaccard-I} on {\it PPI1}. (j) {\sf SimRank-I} vs. {\sf Jaccard-II} on {\it PPI1}.}
    \label{Fig:Difference}
\end{figure*}

\begin{table}[t]
    \centering\scriptsize
    \caption{Differences between {\small \sf SimRank-I} and \break other measures.}
    \begin{tabular}{c|cccc}
        \hline
        {Dataset}	& {Similarity} & {Avg. Bias} & {Max. Bias} & {Min. Bias}\\ \hline
         	 & {\scriptsize \sf SimRank-II} & 0.048 & 0.219 & 0 \\
         {\it Net} & {\scriptsize \sf SimRank-III} & 0.039 & 0.323 & 0 \\
         	 & {\scriptsize \sf Jaccard-I}    & 0.072 & 1 & $2.16 \times 10^{-7}$ \\
             & {\scriptsize \sf Jaccard-II} & 0.160 & 0.770 & 0 \\
        \hline
         	 & {\scriptsize \sf SimRank-II} & 0.075 & 0.609 & 0 \\
        {\it PPI1} & {\scriptsize \sf SimRank-III} & 0.031 & 0.214 & 0 \\
         	 & {\scriptsize \sf Jaccard-I}    & 0.130 & 1 & 0 \\
             & {\scriptsize \sf Jaccard-II} & 0.143 & 0.913 & 0 \\
        \hline
    \end{tabular}
    \label{Tab:Biases}
\end{table}

\noindent\underline{\bf Convergence.} In this experiment, we examined the convergence of our SimRank computation algorithms. We varied the number $n$ of iterations from $1$ to $10$ and computed the SimRank similarities of $1000$ randomly selected vertex pairs by the {\sf Baseline} algorithm. Fig.~\ref{Fig:Converge} shows the effects of $n$ on the average and the maximum SimRank similarities between these $1000$ vertex pairs on {\it PPI1}, {\it PPI2}, {\it Net} and {\it Condmat}. Obviously, the SimRank similarities remain stable after $5$ iterations. This experiment verifies that the approximated SimRank similarity converges to the exact similarity as $n$ becomes larger.

\begin{figure}[h]
    \centering
    \subfigure[Avg. SimRank similarity vs. $n$.]{\includegraphics[width=0.49\columnwidth,height = 1.25in]{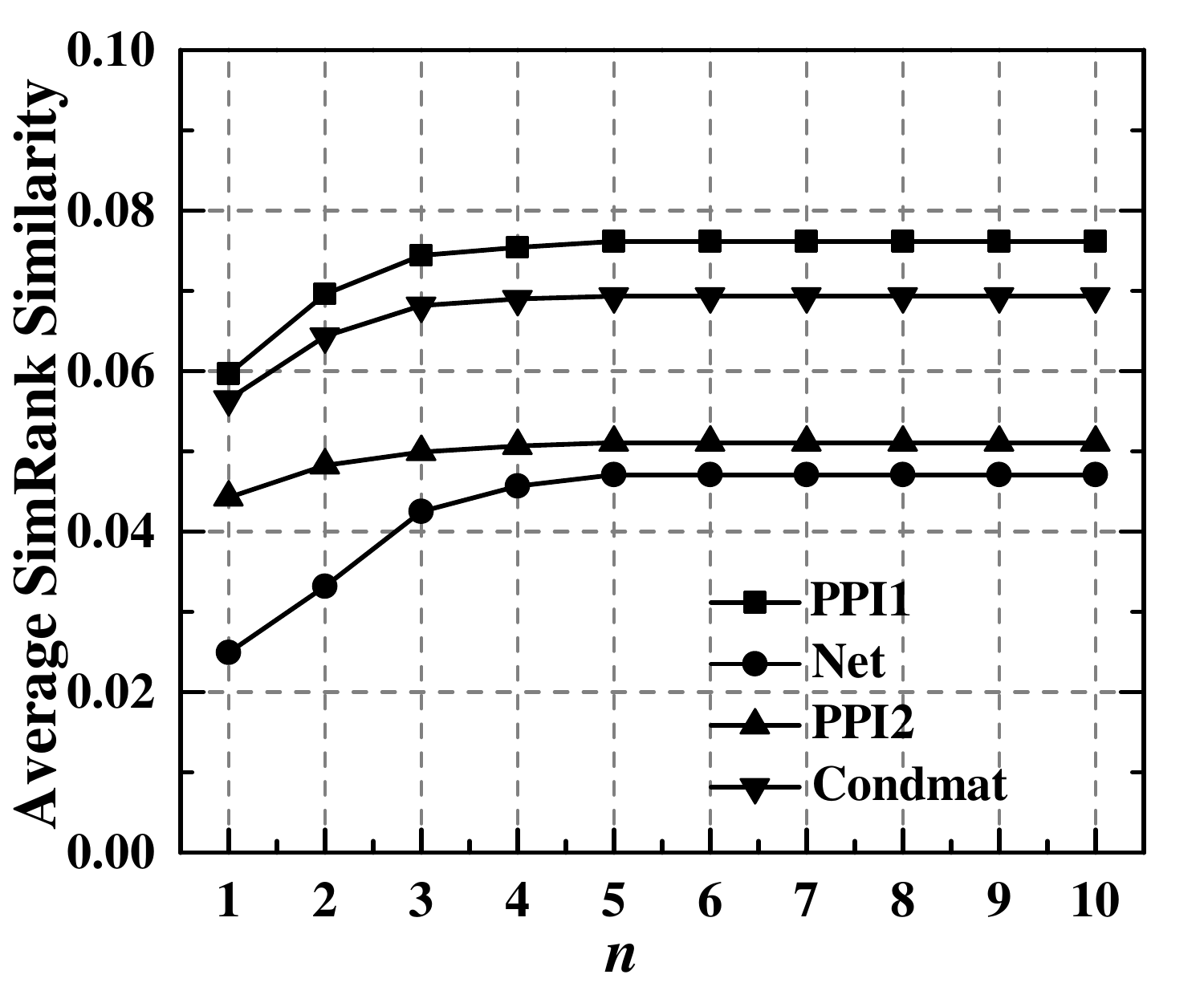}}
    \subfigure[Max. SimRank similarity vs. $n$.]{\includegraphics[width=0.49\columnwidth,height = 1.25in]{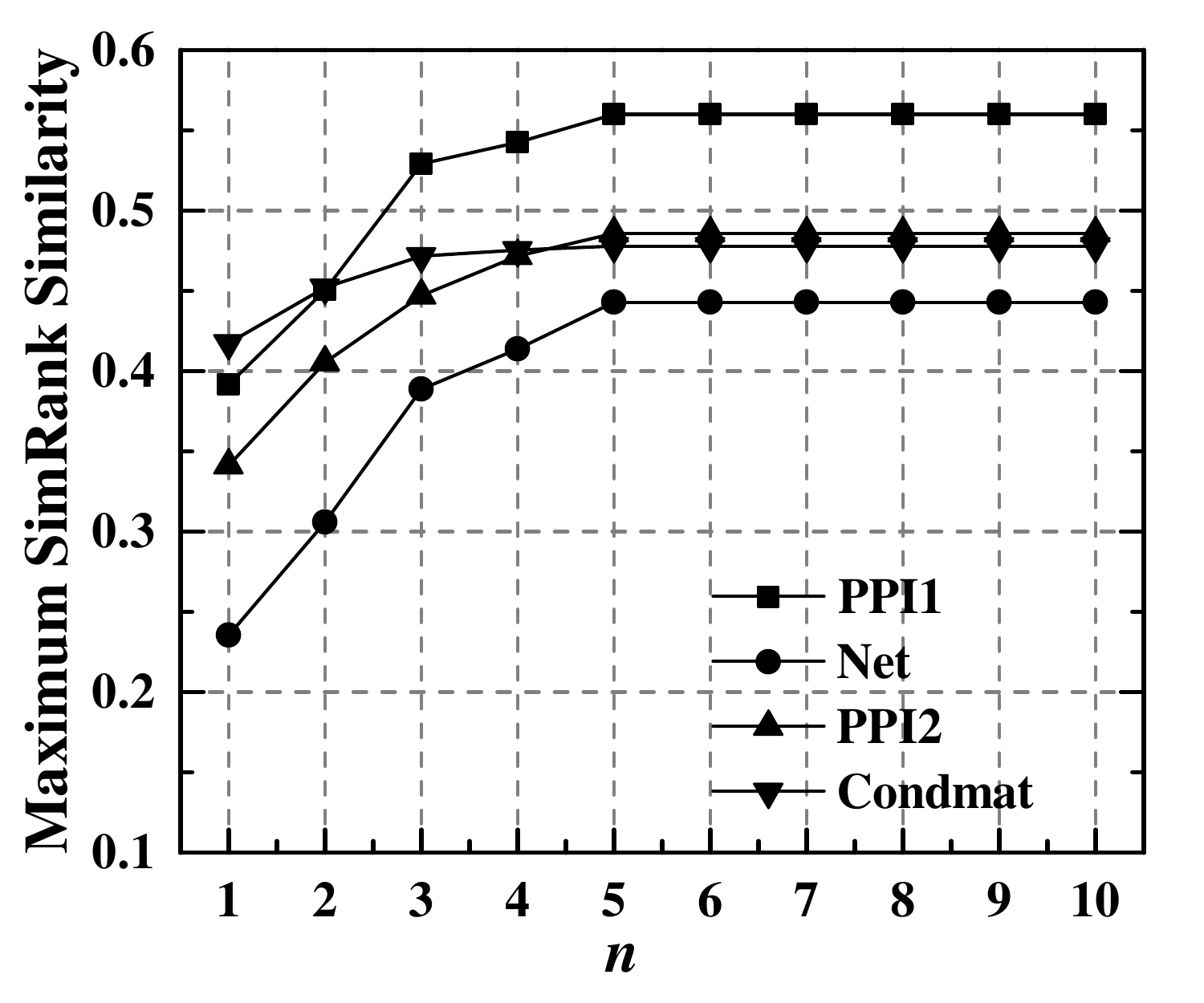}}
    \caption{Effects of the number of iterations on SimRank similarity.}
    \label{Fig:Converge}
\end{figure}

\noindent\underline{\bf Efficiency.} In this experiment, we compared the execution time of the algorithms {\sf Baseline}, {\sf Sampling}, {\sf SR-TS} and {\sf SR-SP}. Fig.~\ref{Fig:Efficiency} illustrates the average execution time of the algorithms. For {\sf SR-TS} and {\sf SR-SP}, we set $l = 1, 2, 3$, respectively. From Fig.~\ref{Fig:Efficiency}, we have the following observations.

1) {\sf Baseline} is faster than {\sf Sampling} and {\sf SR-TS} on {\it PPI2}, {\it PPI3} and {\it Condmat} but is much slower on {\it DBLP}. This is because {\it PPI2}, {\it PPI3} and {\it Condmat} are small graphs, so their transition probability matrices can fit into main memory. However, each of the transition probability matrices $\UW_2, \UW_3, \ldots, \UW_5$ of {\it DBLP} cannot fit into main memory. Thus, {\sf Baseline} incurs high I/O cost on {\it DBLP}.

2) {\sf SR-SP} is much faster than {\sf SR-TS} on all the datasets. The speedup on {\it PPI2} and {\it PPI3} are more than $30$ and $15$, respectively. This is because {\sf SR-SP} uses the speed-up technique, and the bit-wise operations to reduce sampling time.

3) The execution time of {\sf Sampling} is independent of the input graph size. {\sf Sampling} is faster on {\it DBLP} than on {\it PPI2} because the time complexity of {\sf Sampling} is only related to the number of sampled walks and the density of the input graph. The execution time of {\sf Sampling} on {\it PPI3} is very high because {\it PPI3} is very dense.

4) The execution time of {\sf SR-TS} is comparable to {\sf Sampling}. The parameter $l$ of the two-phase algorithm has little effect on the execution time of {\sf SR-TS} because the execution time of {\sf SR-TS} is dominated by the sampling process.

\begin{figure*}[t]
    \centering
    \scriptsize
    \includegraphics[width=\linewidth]{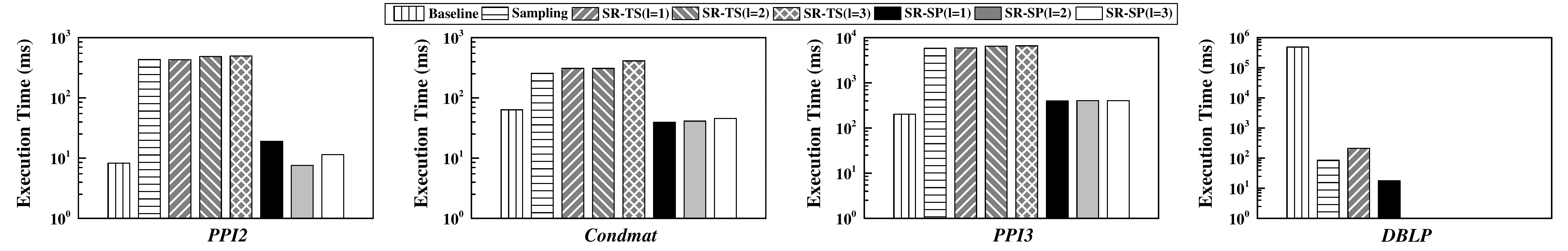}
    \caption{Execution time of algorithms {\sf Baseline}, {\sf Sampling}, {\sf SR-TS} and {\sf SR-SP}. For {\sf SR-TS} and {\sf SR-SP}, we set $l = 1, 2, 3$, respectively.}
    \label{Fig:Efficiency}
\end{figure*}

\begin{figure*}[t]
    \centering
    \scriptsize
    \includegraphics[width=\linewidth]{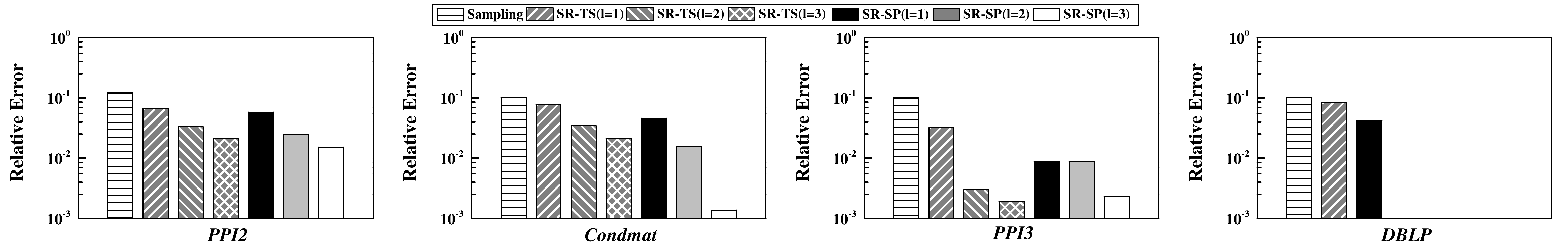}
    \caption{Relative error of algorithms {\sf Sampling}, {\sf SR-TS} and {\sf SR-SP}. For {\sf SR-TS} and {\sf SR-SP}, we set $l = 1, 2, 3$, respectively.}
    \label{Fig:Error}
\end{figure*}

\noindent\underline{\bf Accuracy.}
We examined the accuracy of the algorithms by running them on all datasets $1000$ times and computing the average relative error of the $1000$ runs. Since it is hard to compute the exact SimRank similarity between two vertices, we take the SimRank similarity $s^*$ computed by {\sf Baseline} as the baseline and compute the relative error by $|s - s^*|/s^*$, where $s$ is the similarity computed by a tested algorithm.

Fig.~\ref{Fig:Error} shows the relative errors of the algorithms. We have two observations. 1) The relative errors of the algorithms are very small. In particular, the relative error of {\sf Sampling} is about $10\%$, and the relative errors of {\sf SR-TS} and {\sf SR-SP} are nearly $1\%$. 2) The relative errors of {\sf SR-TS} and {\sf SR-SP} decrease with the growth of parameter $l$. This is consistent with Lemma \ref{Thm:RelBound}. For the same $l$, the relative error of {\sf SR-TS} and {\sf SR-SP} are comparable. Since {\sf SR-TS} is much more efficient than {\sf SR-SP}, {\sf SR-SP} is superior to {\sf SR-TS} in practice.

\noindent\underline{\bf Effects of Parameter $N$.}
In this experiment, we tested the effects of the number $N$ of sampled walks on the efficiency and the accuracy of the algorithms. Since the execution time of {\sf Sampling} is comparable to {\sf SR-TS}, we only tested {\sf SR-TS} and {\sf SR-SP} in our experiment. We set $l = 1$.

\begin{figure}[h]
    \centering
    \subfigure[Execution time vs. $N$.]{\includegraphics[height = 1.2in, width=0.49\columnwidth]{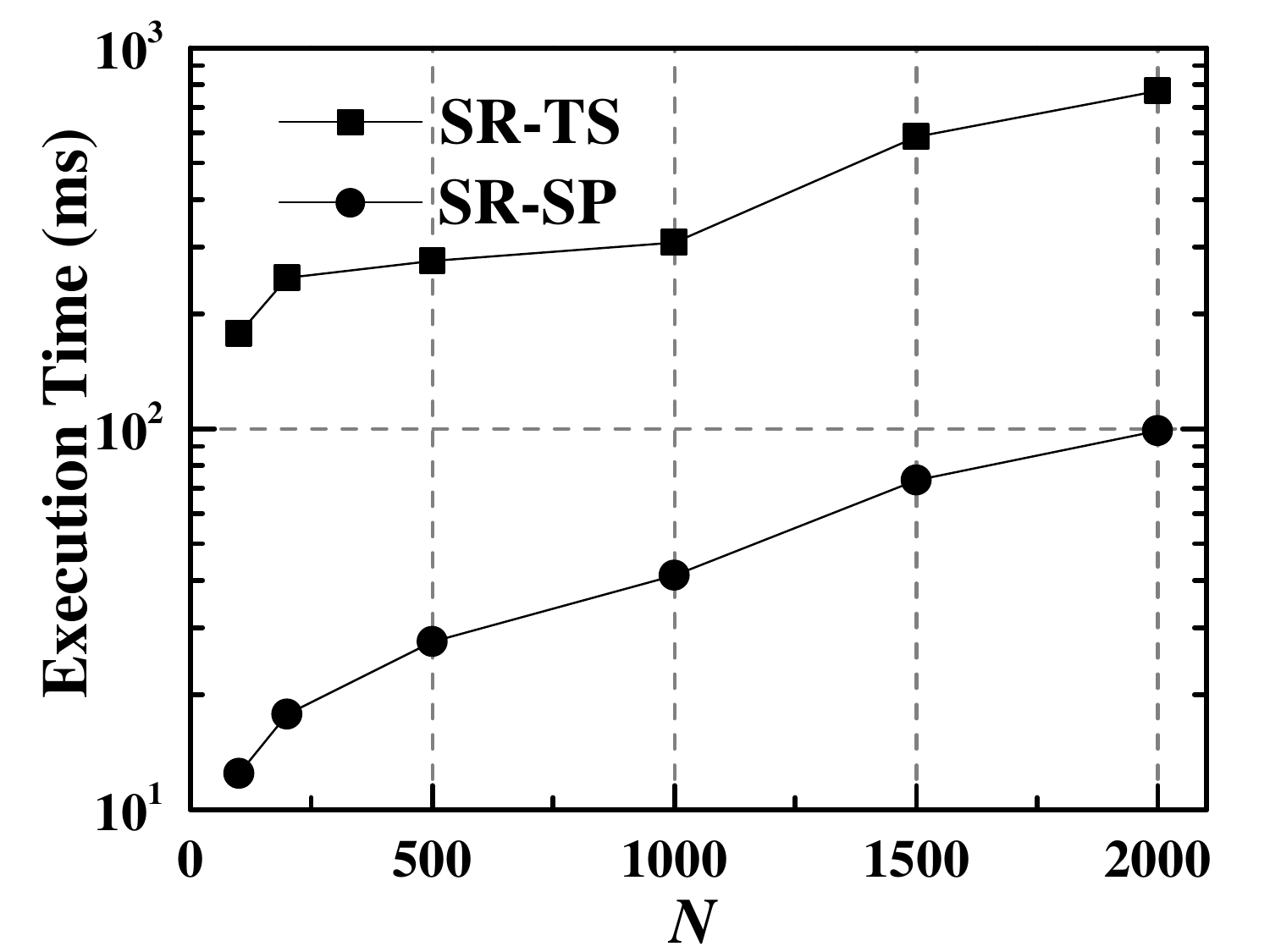}}
    \subfigure[Relative error vs. $N$.]{\includegraphics[height = 1.2in, width=0.49\columnwidth]{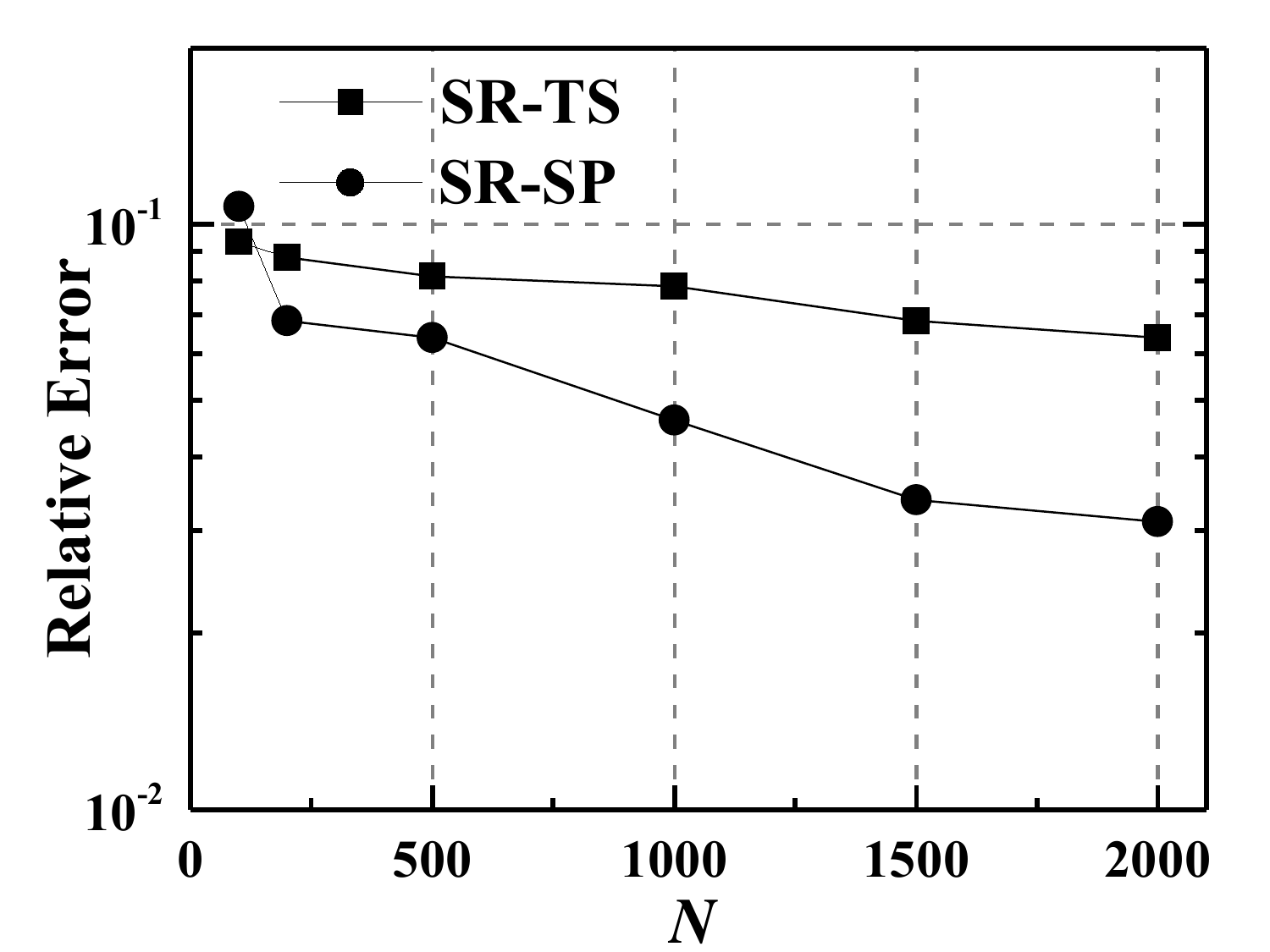}}
    \caption{Effects of parameter $N$ on execution time and relative error.}
    \label{Fig:EffectsN}
\end{figure}

Fig.~\ref{Fig:EffectsN} shows the execution time and the relative error of {\sf SR-TS} and {\sf SR-SP} with respect to $N$ on {\it Condmat}. We observe that the execution time of {\sf SR-TS} and {\sf SR-SP} grow sub-linearly with respect to $N$, and the relative errors decrease with the growth of $N$. When $N$ is sufficient large, the relative error varies slightly. We can observe that the relative error is less than $5\%$ for both algorithms when $N \ge 1000$. The experimental results on the other datasets are similar.

\begin{figure}[h]
\centering
\subfigure{\includegraphics[width=0.48\columnwidth]{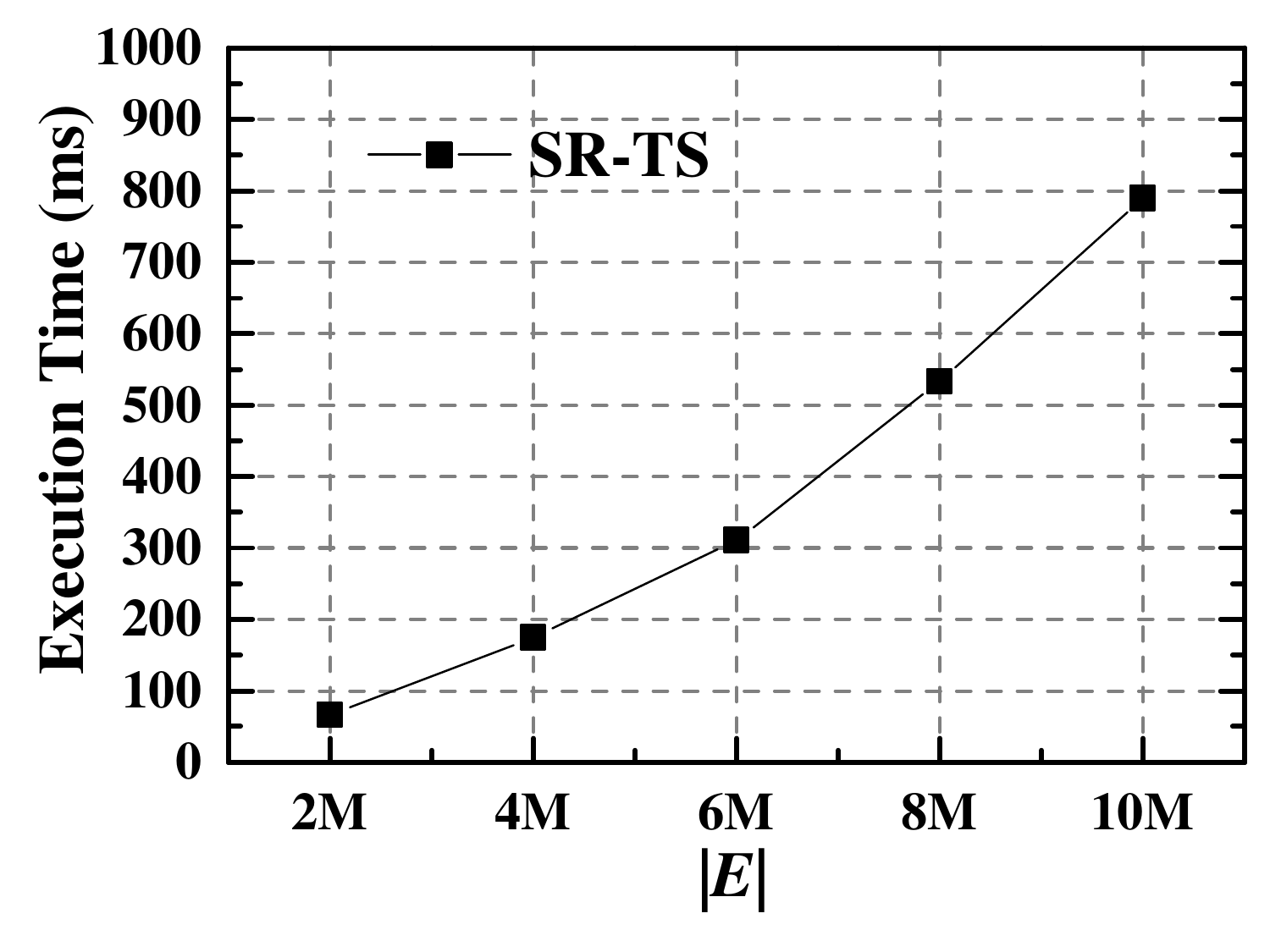}}
\subfigure{\includegraphics[width=0.48\columnwidth]{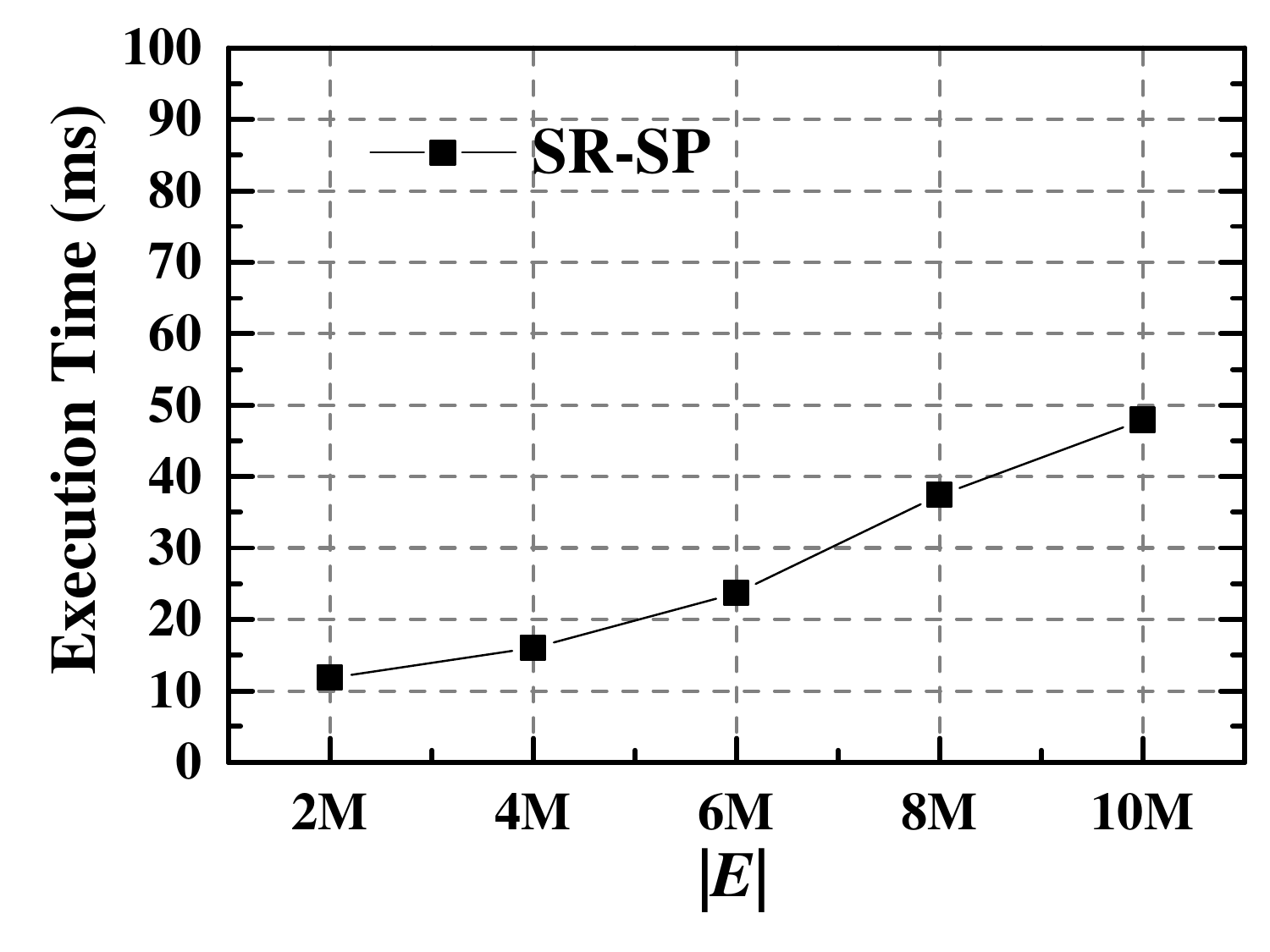}}
\caption{Scalability of {\small \sf SR-TS} and {\small \sf SR-SP} with respect to graph size.}
\label{Fig:Scalability}
\medskip
\end{figure}

\noindent\underline{\bf Scalability.}
In the last experiment, we tested the scalability of the algorithms on synthetic uncertain graphs. We generated a series of uncertain graphs with $2$M vertices and $2$M--$10$M edges. The structures of the uncertain graphs were generated using the R-MAT model \cite{chakrabarti2004r}, and the probabilities of the edges were generated uniformly at random within $[0,1]$. Let $N = 1000$, $n = 5$ and $l = 1$. We ran {\sf SR-TS} and {\sf SR-SP} on each uncertain graph $1000$ times. Fig.~\ref{Fig:Scalability} illustrates the average execution time of {\sf SR-TS} and {\sf SR-SP}. We observe that the execution time of both {\sf SR-TS} and {\sf SR-SP} grow almost linear with the number of edges. This is because the density of the graph is proportional to the number of edges, and the time complexity of {\sf SR-TS} and {\sf SR-SP} are highly dependent on the density of the graph. This experimental results show that our proposed algorithms attain high scalability.

\subsection{Case Study}

We demonstrate two case studies to show the effectiveness of our SimRank similarity measure. One is the detection of proteins with similar biological functions in a protein-protein interaction (PPI) network. The other is the aggregation of similar objects in graph-based entity resolution algorithms.

\noindent\underline{\bf Detecting Similar Proteins.} In this case study, we use our SimRank similarity measure on uncertain graphs to find similar proteins in a PPI network. Here we find the top-$20$ similar protein pairs by two methods. The first one is the SimRank measure proposed in this paper ({\sf USIM}), and the second one is the SimRank measure without considering the uncertainties ({\sf DSIM}) in the PPI network. Fig.~\ref{Fig:TopSimProteins} report the top-$20$ similar protein pairs on the {\em PPI1} dataset. Here we use the MIPS\footnote{\st \tt http://mips.helmholtz-muenchen.de} database as the ground truth, which provides many known protein complexes.  Protein pairs within a common protein complex are thought to coordinate with each other in biological functions \cite{zhao2014detecting, zou2010finding}. The protein pairs in Fig.~\ref{Fig:TopSimProteins} that are contained in the same protein complex are marked in boldface.

\begin{figure}
  \centering
  \scriptsize
  \subfigure[Top-20 similar protein pairs by {\sf USIM}]{
  \begin{tabular}{c} \hline
  {\bf (CTK2,CTK1) (MRPL19, MRPL11) (YHR197W, IPI3)} \\
  {\bf (YPL166W,CIS1) (YDR279W,YLR154C) (NUP85,SEH1)} \\
  {\bf (NUP85,NUP120) (SEH1,NUP120) (BZZ1,VRP1)} (PSF1,CTF4) \\
  {\bf (CTF8,CTF18) (RRN11,RRN7)} (BUL1,YHR131C) \\
  {\bf (NUP85,NUP84) (SEH1,NUP84) (NUP120,NUP84)} \\
  (PHO85,PHO80) {\bf (UBP13,UBP9)} (ESC8 IOC3) (CDC53 GRR1) \\ \hline
  \\
  \end{tabular}
  }
  \subfigure[Top-20 similar protein pairs by {\sf DSIM}]{
  \begin{tabular}{c} \hline
  (YPR143W YKL014C) (YMR067C DOA1) {\bf (YPL166W,CIS1)} \\
  (VIP1 YBR267W) {\bf(YHR197W, IPI3)} {\bf (NUP85,SEH1)}\\
  (IPL1 YDR415C)  (GRH1 SFB2) (LTV1 LYS9) (TUB4 SPC72)\\
  (SAR1 MSK1) {\bf (YDR489W CTF4)} (MBP1 YKR077W)\\
  (ESC8 IOC3) (NUF2 SPC25) {\bf (UBP13,UBP9)} (PHO85,PHO80) \\
  {\bf (MCK1 RIC1)} (VPS28 YGR206W) (YKL014C SSF2) \\ \hline
  \\
  \end{tabular}
  }
  \caption{Top-20 similar Proteins detected by {\sf USIM} and {\sf DSIM}.}
  \label{Fig:TopSimProteins}
\end{figure}

Notice that $16$ pairs of proteins in the top-$20$ results by {\sf USIM} are contained in the same protein complex and the top-$9$ pairs are all in the same protein complex. Whereas, only $6$ pairs of proteins in the top-$20$ results by {\sf DSIM} are verified to be in the same protein complex and only $3$ pairs in the top-$10$ results are in the same protein complex. This comparison results show that our SimRank similarity measure is capable of capturing the structural-context similarity between objects with inherent uncertainties such as the PPI network, which verify the effectiveness of our SimRank measure on uncertain graphs once again.

Also, in Fig.~\ref{Fig:SimilarProteins} we report the top-$5$ similar proteins with respect to a specific protein BUB1, which plays an important role in the mitotic process of a cell. The relationship of BUB1 and RGA1 has been examined in \cite{nelson2007novel}, which claims to find a novel pathway in the mitotic exit where BUB1 coordinates with RGA1 on the spindle.

\begin{figure}[t]
    \centering\scriptsize
    \includegraphics[width=0.75\columnwidth]{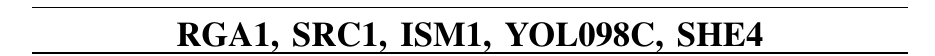}
    \caption{Top-5 Similar protein with respect to protein BUB1.}
    \label{Fig:SimilarProteins}
\end{figure}

\noindent\underline{\bf Entity Resolution.} In this case study, we apply the SimRank similarity measure proposed in this paper to graph-based entity resolution (ER). Following the framework of the {\sf EIF} algorithm \cite{li2010eif}, we develop two new algorithms. The first one is {\sf SimER} that regards the entity graph as an uncertain graph and adopts the SimRank similarity measure proposed in this paper. The second one is {\sf SimDER} that regards the entity graph as an deterministic graph and adopts the SimRank similarity measure on deterministic graphs. We compare the {\sf SimER} algorithm and the {\sf SimDER} algorithm with the {\sf EIF} algorithm \cite{li2010eif} and the {\sf DISTINCT} algorithm  \cite{yin2007object} on the {\em DBLP} dataset. The experimental results are as follows.

First we compare the efficiency of the {\sf SimER} algorithm, the {\sf SimDER} algorithm and the {\sf EIF} algorithm by varying the records size from $2000$ to $5000$. In the {\sf SimER} algorithm and the {\sf SimDER} algorithm, we set the similarity threshold for aggregating data records to be $0.1$. We set the sampling size $N$ to $1000$ and use the speed up techniques in the implementation of {\sf SimER} and {\sf SimDER}. Fig.~\ref{Fig:app1} reports the execution time of the three algorithms. Here the execution time of both the three algorithms all increases approximately linear to the record size because they follow the same framework but only be different on the similarity measures, which are also verified in \cite{li2010eif}. The execution time of the {\sf EIF} algorithm is a bit faster than the {\sf SimER} algorithm and the {\sf SimDER} algorithm. However, the variance is not significant. On average, the {\sf EIF} algorithm is about $20\%$ faster than the {\sf SimDER} algorithm and $25\%$ faster than the {\sf SimER} algorithm. And the {\sf DISTINCT} algorithm is about $25\%$ faster than the {\sf SimDER} algorithm and $30\%$ faster than the {\sf SimER} algorithm.

\begin{figure}[h]
\centering
\includegraphics[width=\columnwidth]{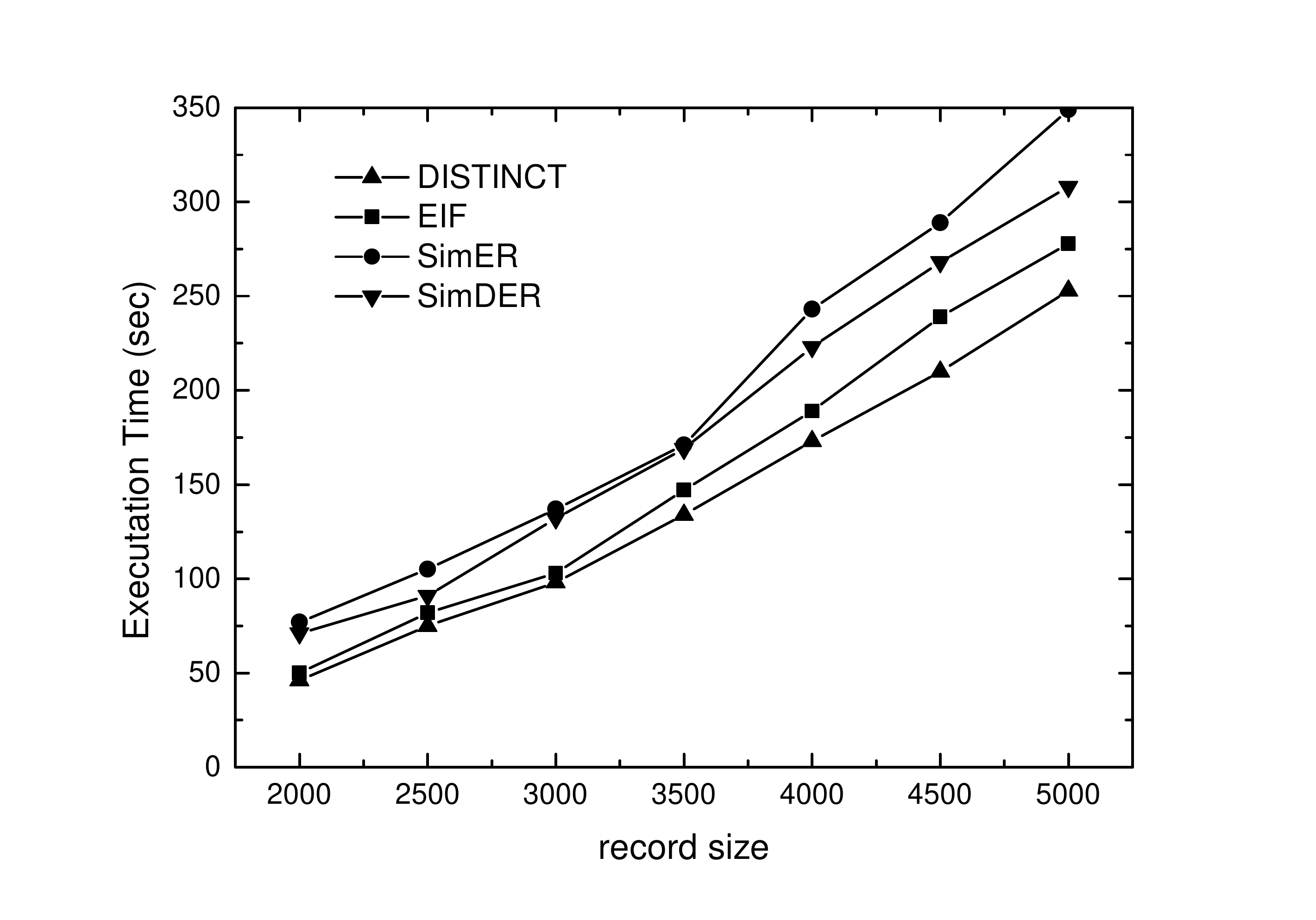}
\caption{Execution time of {\sf EIF}, {\sf SimDER} and {\sf SimER} vs. record size.}
\label{Fig:app1}
\end{figure}

Next we report the effectiveness of the proposed algorithms. We choose the same $8$ representative author names in Table~\ref{Tab:app1} in the {\em DBLP} dataset, where each name corresponds to multiple authors ($\#$authors) and multiple records ($\#$records). We compare the {\it precision}, {\it recall} and {\it F$_1$-measure} of {\sf SimER}, {\sf SimDER}, {\sf EIF} and {\sf DISTINCT}. The experimental results are reported in Table~\ref{Tab:Entity Resolution}.

The {\it precision} of {\sf SimER} and {\sf SimDER} are comparable to that of {\sf EIF} and is slightly better than that of {\sf DISTINCT}. The recall of {\sf SimER} and {\sf SimDER} are much higher than that of {\sf EIF} and {\sf DISTINCT}. Also, {\sf SimER} outperforms {\sf SimDER}, {\sf EIF} and {\sf DISTINCT} in terms of {\it F-measure}. It verifies that the SimRank similarity is an effective method to measure the object similarities in graph data. More over, our SimRank similarity measure which take uncertainties into consideration are more effective for graphs with inherent uncertainties such as the entity resolution application.
\begin{table}
  \centering\scriptsize
  \caption{8 names corresponding to multiple records}
  \begin{tabular}{ccc|ccc}
  \hline
  Name & $\#$authors & $\#$ ref & Name & $\#$authors & $\#$ ref \\ \hline
  Hui Fang & 3 & 9 & Ajay Gupta & 4 & 16 \\ \hline
  Rakesh Kumar & 2 & 38 &  Micheal Wagner & 5 & 24 \\ \hline
  Bing Liu & 6 & 11 & Jim Smith & 3 & 19 \\ \hline
  Wei Wang & 14 & 177 & Wei Wang & 14 & 177 \\ \hline
  \end{tabular}
  \label{Tab:app1}
\end{table}

\begin{table*}[htbp]
    \centering\scriptsize
    \caption{Comparisons of the experimental results of {\sf SimER}, {\sf SimDER}, {\sf EIF} and {\sf DISTINCT} on multiple author names.}
    \begin{tabular}{ccccccccccccccc}
    	\hline
    	Name &  & {\sf SimER} & & &{\sf SimDER}& & &{\sf EIF} & & &{\sf DISTINCT}&  \\
         & {\it Precision}  & {\it Recall} & {\it F-Measure} & {\it Precision} & {\it Recall} & {\it F-Measure} & {\it Precision} & {\it Recall} & {\it F-Measure} & {\it Precision} & {\it Recall} & {\it F-Measure}\\
    	\hline
    	Hui Fang &  1.0 & 1.0 &1.0 & 1.0 & 1.0 & 1.0 & 1.0 & 1.0 & 1.0 & 1.0 & 1.0 & 1.0 \\
        Ajay Gupta &  1.0 &	1.0 & 1.0 & 1.0 & 0.927 & 0.962 & 1.0 & 0.882	& 0.937 & 1.0 &	1.0 & 1.0 \\
        Rakesh Kumar &  1.0 & 1.0 &1.0 & 0.958 & 0.973 & 0.965 & 1.0 & 1.0 & 1.0 & 1.0 & 1.0 & 1.0 \\
        Micheal Wagner &  0.993 & 0.880 & 0.933 & 1.0 & 0.859 & 0.924 & 1.0 & 0.620 &	0.765 & 1.0 & 0.395 & 0.566 \\
        Bing Liu &  1.0 & 1.0 & 1.0 & 0.982 & 1.0 & 0.991 & 1.0 & 1.0 & 1.0 &	1.0 & 0.825 & 0.904 \\
        Jim Smith &  0.988 & 0.963 & 0.975 & 0.973 & 0.978 & 0.975 & 1.0 & 0.810 & 0.895 & 0.888 & 0.926 & 0.907 \\
        Wei Wang &  1.0 & 1.0 &	1.0 & 0.993 & 1.0 & 0.996 & 1.0 & 0.933 & 0.965 & 0.855 & 0.814 & 0.834 \\
        Bin Yu &  0.980 & 0.975 &	0.989 & 0.989 & 0.752 & 0.854 &	0.977 & 0.595 & 0.746 & 1.0 &	0.658 & 0.794\\
        \hline
        {\bf Average} &  \bf 0.995 & \bf 0.977 & \bf 0.987 &  \bf 0.987 & \bf 0.936 & \bf 0.958 & \bf 1.0 & \bf 0.855	& \bf 0.914	 & \bf 0.968 & \bf 0.827 & \bf 0.876 \\
        \hline
    \end{tabular}
    \label{Tab:Entity Resolution}
\end{table*}

\section{Related Work}
\label{sec:related}

\subsection{SimRank on Deterministic Graphs}

SimRank \cite{jeh2002simrank} is a measure of similarities between vertices. A large number of algorithms have been proposed to compute SimRank similarities \cite{
fogaras2005scaling,fujiwara2013efficient,kusumoto2014scalable,lee2012top,li2010fast, lizorkin2008accuracy, shao2015efficient,tao2014efficient, yu2013towards, yu2014fast, yu2015efficient, yu2012space}. These algorithms can be classified into three categories.

\noindent\underline{\bf Iterative Algorithms.} Jeh and Widom \cite{jeh2002simrank} proposed the first iterative algorithm for computing SimRank similarities by matrix computations. It computes the SimRank similarities between all pairs of vertices in $O(n d^2 |V|^2)$ time, where $n$ is the number of iterations, $d$ is the average degree of vertices, and $|V|$ is the number of vertices. Later, \cite{lizorkin2008accuracy} improved the running time to $O(nd|V|^2)$ by sum memorization, and \cite{yu2012space} speeded up the algorithm by fast matrix multiplication. Further, \cite{yu2013towards} decreased the time complexity to $O(nd'|V|^2)$ by fine-grained memorization, where $d' < d$. For the single-pair SimRank computation problem,  \cite{li2010fastSDM} computes the SimRank similarity between a pair of vertices in $O(nd^{2}\min\{|V|^2, d^n\})$ time, and \cite{he2014assessing} further improved the running time to $O(n|E|^2 - |E|)$ by employing position probabilities, where $|E|$ is the number of edges. All these iterative algorithms require $O(|V|^2)$ space and thus can not scale to large graphs. Recently, \cite{yu2015efficient} developed a space-efficient algorithm to compute the SimRank similarities between a subset of vertices in $O(|E|+n|V|)$ space.

\noindent\underline{\bf Non-iterative Algorithms.} In \cite{li2010fast}, a non-iterative SimRank matrix formula has been established based on the Kronecker product and singular vector decomposition (SVD). It pre-computes several auxiliary matrices offline in $O(r^4 |V|^2)$ time and then retrieves the SimRank similarities between a given vertex and all other vertices online in $O(r^4 |V|)$ time, where $r$ is the rank of the adjacency matrix of the graph. Moreover, this algorithm is capable of processing dynamic graphs.  Later, \cite{fujiwara2013efficient} and \cite{yu2014fast} improved the running time to $O(r|V|)$ by Sylvester's equation. Recently, \cite{maehara2015scalable} proposed a novel approach to converting the SimRank formula into a linear equation. The algorithm requires only $O(|E|)$ space and computes all-pairs SimRank similarities and single-pair SimRank similarities in $O(n|V||E|)$ time and $O(n|E|)$ time, respectively.

\noindent\underline{\bf Random-walk-based Algorithms.} The SimRank similarity between two vertices $u$ and $v$ can be represented in form of the probabilities that two random walks starting from $u$ and $v$, respectively, meet at the same vertex at the same time. The first random-walk-based algorithm \cite{fogaras2005scaling} stores the fingerprints of $N$ random walks as an index in $O(N|V|)$ space and computes single-pair SimRank similarities based on this index. In \cite{kusumoto2014scalable}, a sampling algorithm was developed to compute the SimRank similarity between a single pair of vertices based on the random-walk interpretation of the linear formula given in \cite{maehara2015scalable}. For top-$k$ SimRank computations, the algorithm \cite{lee2012top} finds the $k$ vertices with the highest SimRank similarities with respect to a given vertex. The studies in \cite{shao2015efficient} improve the performance by compactly storing some samples of the origin graph. Different from this, \cite{tao2014efficient} studied the problem of finding the top-$k$ pairs of vertices with the highest SimRank similarities.

\subsection{Random Walks on Uncertain Graphs}

Although the concept of random walks on uncertain graphs has ever been used earlier in the literature \cite{du2015probabilistic,jin2011distance,potamias2010kNN}, it is totally different from our definition in this paper. In \cite{du2015probabilistic,jin2011distance,potamias2010kNN}, for a random walk that is on vertex $u$ at time $t$, they sample the neighbors of $u$, randomly select a sampled neighbor $v$, and transit the state from $u$ to $v$ at time $t + 1$. Therefore, for a vertex $u$ and two different time $t$ and $t'$, the possible vertices that the walk can transit to from vertex $u$ at time $t$ are different from those that the walk can transit to from vertex $u$ at time $t'$. However, on each possible world of the uncertain graph, the set of possible vertices that the walk can transit to from vertex $u$ are the same all the time. Thus, the random walk \cite{du2015probabilistic,jin2011distance,potamias2010kNN} does not satisfy Markov's Property. In fact, our definition of random walks on uncertain graphs is the first one that satisfies Markov's Property (see Section \ref{Sec:UncertainRandomWalks}).

\subsection{Similarities for Uncertain Graphs}
Besides the algorithms introduced above, a number of variants of SimRank have also been investigated. \cite{antonellis2008simrank++, jin2011axiomatic, sun2011link, yu2013more, zhao2009p}. However, almost all the studies on SimRank overlook uncertainty in graphs. In recent years, many studies on managing and mining uncertain graphs \cite{jin2011discovering, kollios2013clustering, zou2009frequent,zou2010finding} suggested that the effects of uncertainty on the quality of query or mining results have been undervalued. To the best of our knowledge, the only work on the SimRank on uncertain graphs has only been carried out in \cite{du2015probabilistic}. However, this work made an unreasonable assumption that the $k$-step transition probability matrix $\UW^{(k)}$ equals to the $k$th power of the one-step transition probability matrix $\UW^{(1)}$, i.e., $\UW^{(k)} = (\UW^{(1)})^k$ for all $k \ge 1$. This assumption is inconsistent with the possible world model of uncertain graphs that has been widely adopted in the literature.

Besides the SimRank similarity measure, three structural-context similarity measures between vertices of an uncertain graph have been proposed \cite{zou2013structural}, namely the expected Jaccard similarity, the expected Dice similarity and the expected cosine similarity. Unlike the SimRank similarity, these three measures are only applicable to the vertices with common neighbors.

\section{Conclusions}

This paper proposes the concepts and the algorithms related to the SimRank similarity on uncertain graphs. Unlike the most related work \cite{du2015probabilistic}, our concepts are completely based on the possible world model of uncertain graphs. We propose three algorithms and a speeding-up technique for SimRank computation on an uncertain graph. The experimental results show that the algorithms are effective and efficient. To lay foundations for SimRank, we also study random walks on uncertain graphs for the first time. We reveal the critical differences between random walks on uncertain graphs and the counterparts on deterministic graphs.

\nocite{*}

\section*{Acknowledgements}
We thank anonymous reviewers for their very useful comments and suggestions.

This work was partially supported by the 973 Program of China under grant No. 2011CB036202 and the National Natural Science Foundation of China under grant No. 61173023 and No.61532015.

\scriptsize
\bibliographystyle{abbrv}
\bibliography{conf_icde_full}

\end{document}